\documentclass[a4paper,UKenglish,cleveref, autoref, thm-restate]{lipics-v2021}

\pdfoutput=1 %
\hideLIPIcs  %

    \title{Maximizing Reachability via Shifting of Temporal Paths}

\author{Argyrios Deligkas}
    {Royal Holloway, University of London, Egham, United Kingdom \and \url{https://sites.google.com/view/deligkas/home}}
    {argyrios.deligkas@rhul.ac.uk}
    {https://orcid.org/0000-0002-6513-6748}
    {EPSRC Grant EP/X039862/1 ``NAfANE: New Approaches for Approximate Nash Equilibria''}

    \author{Michelle D\"oring}
    {Hasso Plattner Institute, University of Potsdam, Potsdam, Germany \and \url{https://www.notion.so/Michelle-D-ring-1dd43d6e8b7c800eadbdd32d73e21b72?pvs=25}}
    {michelle.doering@hpi.de}
    {https://orcid.org/0000-0001-7737-3903}
    {HPI Research School on Foundations of AI (FAI)}

    \author{Eduard Eiben}
    {Royal Holloway, University of London, Egham, United Kingdom \and \url{https://pure.royalholloway.ac.uk/en/persons/eduard-eiben}}
    {eduard.eiben@rhul.ac.uk}
    {https://orcid.org/0000-0003-2628-3435}
    {}

    \author{George Skretas}
    {Hasso Plattner Institute, University of Potsdam, Potsdam, Germany}
    {georgios.skretas@hpi.de}
    {https://orcid.org/0000-0003-2514-8004}
    {}

    \author{Georg Tennigkeit}
    {Hasso Plattner Institute, University of Potsdam, Potsdam, Germany}
    {georg.tennigkeit@hpi.de}
    {https://orcid.org/0000-0003-0734-0684}
    {HPI Research School on Foundations of AI (FAI)}

\authorrunning{A. Deligkas, M. Döring, E. Eiben, G. Skretas and G. Tennigkeit} %

\Copyright{Jane Open Access and Joan R. Public} %



\category{} %

\relatedversion{} %

\nolinenumbers %

\EventEditors{John Q. Open and Joan R. Access}
\EventNoEds{2}
\EventLongTitle{42nd Conference on Very Important Topics (CVIT 2016)}
\EventShortTitle{CVIT 2016}
\EventAcronym{CVIT}
\EventYear{2016}
\EventDate{December 24--27, 2016}
\EventLocation{Little Whinging, United Kingdom}
\EventLogo{}
\SeriesVolume{42}
\ArticleNo{23}

\usepackage[dvipsnames]{xcolor}
\usepackage{xspace}
\usepackage{fancyhdr}
\usepackage{tcolorbox}
\usepackage{graphicx}
\usepackage{enumitem}

\usepackage{amssymb}
\usepackage{amsthm}

\usepackage{amsmath,mathtools}
\usepackage{bbm}

\usepackage{stmaryrd} %
\usepackage{stackengine}
\usepackage[english,ruled,vlined,algo2e,linesnumbered]{algorithm2e}
\usepackage{bm}
\usepackage{thm-restate}

\usepackage{multirow}
\usepackage{tabularray}
\usepackage{colortbl}

\usepackage{anyfontsize}
\usepackage[percent]{overpic}

\definecolor{darkgray}{gray}{0.25}
    \newenvironment{construction*}{%
        \par\vspace{0.25\baselineskip}%
        \pushQED{\qed}%
        \noindent\textcolor{darkgray}{$\triangleright$}\;%
        \noindent\textbf{\textcolor{darkgray}{Construction~\theconstruction.}}\ %
    }{%
      \popQED\par\vspace{0.25\baselineskip}%
    }
\theoremstyle{remark}

\declaretheoremstyle[
  headfont=\bfseries,
  headformat=$\lozenge$\ \NAME\NOTE, %
  headpunct={.},
  notefont=\bfseries,
  bodyfont=\normalfont,
  spaceabove=6pt, spacebelow=6pt
]{cleancons}
\theoremstyle{cleancons}

\makeatletter
\newcommand{\namedlabel}[2]{\def\@currentlabel{#2}\label{#1}}
\makeatother

\newcommand{\ie}{i.\,e.,\xspace}

\newcommand{\scal}{\ensuremath{\mathcal{S}}\xspace}
\newcommand{\rcal}{\ensuremath{\mathcal{R}}\xspace}

\newcommand{\pcal}{\ensuremath{\mathcal{P}}\xspace}
\newcommand{\gcal}{\ensuremath{\mathcal{G}}\xspace}
\newcommand{\ecal}{\ensuremath{\mathcal{E}}\xspace}
\newcommand{\tcal}{\ensuremath{\mathcal{T}}\xspace}

\newcommand{\vcal}{\ensuremath{\mathcal{V}}\xspace}
\newcommand{\ccal}{\ensuremath{\mathcal{C}}\xspace}

\newcommand{\Na}{\mathbb{N}}
\newcommand{\Int}{\mathbb{Z}}

\newcommand{\tuple}[1]{\ensuremath{\langle {#1} \rangle}\xspace}

\newcommand{\bigparagraph}[1]{\vspace{0.4em}\noindent\textbf{#1}}

\newcommand{\lifetime}{\ensuremath{\Lambda}\xspace}

\newcommand{\kpathgraph}[1]{\ensuremath{#1}-path graph\xspace}

\newcommand{\kpathgraphs}[1]{\ensuremath{#1}-path graphs\xspace}
\newcommand{\kPathGraphs}[1]{\ensuremath{#1}-Path Graphs\xspace}

\newcommand{\dprm}{\textsc{MR-DP}\xspace}
\newcommand{\DPRM}{\textsc{MaxReach-DelayPath}\xspace}
\newcommand{\aprm}{\textsc{MR-AP}\xspace}
\newcommand{\APRM}{\textsc{MaxReach-AdvancePath}\xspace}
\newcommand{\sprm}{\textsc{MR-SP}\xspace}
\newcommand{\SPRM}{\textsc{MaxReach-ShiftPath}\xspace}
\newcommand{\dasprm}{\textsc{MR-D/A/SP}\xspace}

\newcommand{\rprm}{\textsc{MR-PT}\xspace}
\newcommand{\RPRM}{\textsc{MaxReach-PathTemporalization}\xspace}

\newcommand{\subpath}[3]{\ensuremath{#1[#2:#3]}\xspace}
\newcommand{\suffix}[2]{\ensuremath{#1[#2:]}\xspace}
\newcommand{\prefix}[2]{\ensuremath{#1[:#2]}\xspace}

\newcommand{\svset}{\vcal}
\newcommand{\sptree}{\tcal}

\newcommand{\switchvertexset}{\ensuremath{\mathsf{SVS}}\xspace}
\newcommand{\switchvertexsets}{\ensuremath{\mathsf{SVSs}}\xspace}
\newcommand{\switchpathtree}{\ensuremath{\mathsf{SPT}}\xspace}
\newcommand{\switchpathtrees}{\ensuremath{\mathsf{SPTs}}\xspace}

\newcommand{\switchonto}[1]{\ensuremath{v^{#1}_\svset}\xspace}

\newcommand{\sw}{\ensuremath{\mathsf{sw}}\xspace}
\newcommand{\swEto}[1]{\ensuremath{e^{#1}_{to}}\xspace}
\newcommand{\swEfrom}[1]{\ensuremath{e^{#1}_{from}}\xspace}
\newcommand{\swPto}{\ensuremath{P_{to}}\xspace}
\newcommand{\swPfrom}{\ensuremath{P_{from}}\xspace}

\newcommand{\NP}{\ensuremath{\mathtt{NP}}\xspace}

\newcommand{\Wone}{\ensuremath{\mathtt{W}[1]}\xspace}
\newcommand{\Wtwo}{\ensuremath{\mathtt{W}[2]}\xspace}

\newcommand{\FPT}{\ensuremath{\mathtt{FPT}}\xspace}
\newcommand{\XP}{\ensuremath{\mathtt{XP}}\xspace}

\newcommand{\poly}{\ensuremath{\mathtt{poly}}\xspace}

\newcommand{\bigoh}{\mathcal{O}}

\newcommand{\MCIS}{\textsc{Multi-Colored Independent Set}\xspace}
\newcommand{\mcis}{\textsc{MCIS}\xspace}

\usepackage{tabularx}
\usepackage{graphicx}   %
\usepackage{hhline}     %
\usepackage{multirow}   %
\newcommand{\problemtitle}[1]{\gdef\@problemtitle{#1}}%
\newcommand{\probleminput}[1]{\gdef\@probleminput{#1}}%
\newcommand{\problemquestion}[1]{\gdef\@problemquestion{#1}}%
\newif\iflong
\newif\ifshort
\longtrue

\iflong
\else
\shorttrue
\fi

\usepackage{todonotes}

\newenvironment{tightcenter}
 {\parskip=0pt\par\nopagebreak\centering}
 {\par\noindent\ignorespacesafterend}

\usepackage{tikz}
\usetikzlibrary{calc}
\usepackage{xargs}
\usepackage{xifthen}
\usepackage{framed}

\usepackage{ctable}

\newlength{\RoundedBoxWidth}
\newsavebox{\GrayRoundedBox}
\newenvironment{GrayBox}[1]%
   {\setlength{\RoundedBoxWidth}{\textwidth-10.5ex}
    \def\boxheading{#1}
    \begin{lrbox}{\GrayRoundedBox}
       \begin{minipage}{\RoundedBoxWidth}%
   }{%
       \end{minipage}
    \end{lrbox}%
    \begin{tightcenter}%
    \begin{tikzpicture}%
       \node(Text)[draw=black!90,fill=white,rounded corners,%
             inner sep=2ex,text width=\RoundedBoxWidth]%
             {\usebox{\GrayRoundedBox}};
        \coordinate(x) at (current bounding box.north west);
        \node [draw=white,rectangle,inner sep=3pt,anchor=north west,fill=white] 
        at ($(x)+(6pt,.75em)$) {\boxheading};
    \end{tikzpicture}
    \end{tightcenter}\vspace{0pt}%
    \ignorespacesafterend
}    

\newenvironment{problem}[2][]{\noindent\ignorespaces%
                                \FrameSep=8pt%
                                \parindent=0pt%
                \vspace*{-.5em}
                \ifthenelse{\isempty{#1}}{%
                  \begin{GrayBox}{\textsc{#2}}%
                }{%
                }
                \newcommand\Prob{{Problem:}}%
                \newcommand\Input{{Input:}}%
                          
                \begin{tabular*}{\textwidth}{@{\hspace{.1em}} >{\itshape} p{1.2cm} p{0.85\textwidth} @{}}%
            }{
                \end{tabular*}%
                \end{GrayBox}%
                \vspace*{-.5em}
                \ignorespacesafterend
            }

\begin{document}

\maketitle

\begin{abstract}
We examine the problem of maximizing the reachability of a given source in temporal graphs
that are given as the union of $k$ temporal paths, i.e., every given path is a sequence of edges with strictly increasing labels that denote availability in time.
This type of temporal graphs represent train networks.
We consider shifting operations on the labels of the paths that maintain their temporal continuity. %
This means that we can move the availability of a temporal edge later %
or earlier 
in time, and propagate the shifts to all other affected edges of the path in order to preserve its temporal connectivity.
We study the parameterized complexity of the problem with respect to the number of paths $k$, and the total budget $b$, where $b$ is the maximum number of shifts we are allowed to perform.
Our results reveal that fixed parameter tractability can be achieved
(1) when parameterized both by $k$ and $b$, and  
(2) when parameterized by $k$, and $b$ is unconstrained.
In almost every other case, e.g., parameterized by a single parameter or parameterized by $k$, while having a bound on $b$, we establish intractability lower bounds that are matched by XP algorithms.
\end{abstract}

\ifshort
\vspace{2em}
\noindent Due to space limitations, we deferred all formal proofs and many helpful illustrations of the statements to the full version.
\fi

\newpage

\section{Introduction}

Probably the simplest way to describe the problem of maximizing reachability on a temporal graph is via a train network.
We are given the routes and corresponding timings for each train, along with a specific station $s$. 
The question is, how should we {\em modify} the network such that $s$ can reach as many other stations as possible?

Formally, in a temporal graph, every edge has a set of {\em labels} that correspond to the specific timesteps at which the edge can be used. 
The modifications of the temporal graph proposed in the past ranged from the addition of (non-existing) edges~\cite{bellitto_temporalconnectivity_2025} and labels, 
to {\em advancing} or {\em delaying}, i.e., {\em shifting}, the labels of the edges~\cite{deligkas_minimizingreachability_2023,deligkas_optimizingreachability_2022,kutner_betterlate_2025,molter_temporalreachability_2024}.
In real-life networks, though, especially physical ones, it is almost impossible to add new edges. We cannot simply add one direct connection from Warsaw to L'Aquila because we have decided so. 
Hence, shifting seems a more suitable operation in applications.

On the other hand, in certain scenarios, like train networks, delaying or advancing a temporal edge cannot happen in isolation. 
Trains pass through many different stations, and we cannot simply change the time at which it goes between two consecutive stations without potentially affecting the schedule of the entire train line. 
Imagine, for example, a train that leaves from $A$ to $B$ at 9:00 and from $B$ to $C$ at 10:00. We cannot delay the $A$-$B$ connection to 12:00 and maintain the $B$-$C$ connection at 10:00; the delay {\em must propagate} and connection $B$-$C$ {\em must} be after 13:00.
Thus, changes on a train line that occur due to shifting must propagate until feasibility on the corresponding temporal path is maintained. 

While maximum reachability via delaying operations had been studied in the past, the above-mentioned dimension of real life networks was overlooked. The goal of the paper is remedy this situation and resolve the (parameterized) complexity of maximum reachability of a vertex via shifting on temporal $k$-path graphs, i.e., temporal graphs that are defined by a collection of $k$ temporal base-paths, denoted \SPRM (\sprm for short).

\begin{problem}[]{{ \SPRM (\sprm)}}
    \Input &A \kpathgraph{k} $\gcal$, a budget $b\in \Na$, and a vertex $s$.\\
    \Prob & {Find a sequence of shifting operations with a total cost of at most $b$, such that the number of vertices that $s$ can reach in the resulting $\gcal'$ is maximized.}
\end{problem}

More formally, a shifting operation either increases the time label of a specified edge (delaying) or decreases it (advancing). 
When delaying an edge, the delay propagates to subsequent edges on the unique containing base-path insofar as the labels remain strictly increasing. 
Likewise, advancing an edge propagates to preceding edges on the containing base-path. 
If there was any waiting time between edges in the base-path, it reduces the propagating shift, as if the delayed train can make up for its delay by driving faster on the connection, or waiting less at the station. 

The {\em cost} of one shifting operation is equal to the absolute value by which the specified edge is modified, while propagated changes do not count towards the cost. This reflects that the budget measures only the \emph{active intervention} by the operator, namely how much a single departure is intentionally rescheduled. The resulting propagation along the same base-path is then a forced consequence of preserving a feasible timetable for that line, rather than an additional independent modification.

Figure~\ref{fig:intro-example} demonstrates an instance of \sprm alongside a solution.

\begin{figure}[t]
    \centering
    \includegraphics[width=0.65\linewidth]{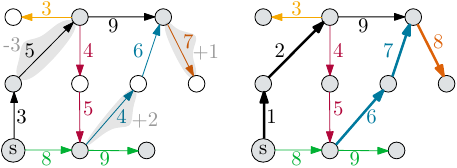}
    \caption{An example instance of \sprm on a temporal 6-path graph, whose base-paths are depicted with different colours,  along with a sequence of operations that allows $s$ to reach all vertices for a cost of 6. Gray vertices are reached by $s$.
   }
    \label{fig:intro-example}
\end{figure}

\subsection{Our Contribution}
We begin our investigation by establishing some observations and deriving structural properties of optimal solutions in \kpathgraphs{k} that help us design our algorithms.
The most fundamental observation, that our algorithms heavily exploit, is the fact that in any \kpathgraph{k} if vertex $u$ reaches vertex $v$, then there is a temporal path from $u$ to $v$ that uses {\em at most one} contiguous segment from each base-path.
This implies that there exists an optimal solution where each base-path {\em switches} to a different base-path at most once.
This in turn implies the existence of a {\em switch-path-tree}, denoted \switchpathtree. This is a directed tree, rooted at the source, which describes how temporal paths starting from the source should switch in order to reach different base-paths.
Crucially, we prove that we can enumerate all \switchpathtrees of a \kpathgraph{k} in $\bigoh\!\left(k^{k-2}\right)$ time (\Cref{lem:spt_enumeration}).

Then, we continue by studying the complexity of our problem with respect to two dimensions: a) the parameters we consider and b) the types of shifting operations we are allowed to perform, i.e., delay, advance, or both. 
We begin our exploration from the most relaxed version of the problem where we allow both delays and advances, and the budget is unconstrained.
Although the problem is \NP-hard (\Cref{thm:rprm_NP-hard}), when parameterized by the number $k$ of base-paths we prove that it becomes \FPT (\Cref{thm:rprm_FPT_by_k}). 
At a high level, out algorithm guesses the \switchpathtree of an optimal solution and then executes a suitable modification of BFS to compute the shifts that maximize reachability.

In the general setting where we must adhere to a budget $b$, we prove that if we allow only delays, then the problem becomes \Wone-hard parameterized by $k$ (\Cref{thm:delay_w1_by_k}).
On the other hand, if we allow advances, or both advances and delays, the complexity of the problem is partially answered: the problem is \Wone-hard if we ask to find the optimal shifts for a {\em fixed} \switchpathtree (\Cref{thm:fixedSPT_w1_by_k}) and open otherwise.
On the positive side though, we match all above-mentioned hardness results with upper bounds by designing an \XP-algorithm that tackles all versions of the problem (\Cref{thm:budgeted_XP_by_k}). This algorithm guesses the essential switching points of an optimal solution and then uses an integer linear program (ILP) to find the minimum cost needed to realize them. The trick is to precompute the effects of propagation to keep the ILP sufficiently small.

Then, we consider the complexity of the problem when the budget $b$ is a parameter. Unfortunately, we observe that the result of~\cite{deligkas_minimizingreachability_2023} essentially provides \Wtwo-hardness for all types of shifts.
On the other hand, we complement these lower bounds by observing that there exists a straightforward \XP algorithm (\Cref{thm:budgeted_XP_by_b}).
Hence, in order to establish tractability we have to augment our set of parameters. 
Indeed, as we show in \Cref{thm:delay_fpt_by_b_and_k} and \Cref{thm:all_fpt_by_b_and_k} the problem becomes fixed parameter by $k$ and $b$ tractable for every allowed operation of shifts.
This are our most technical results, which works at a very high level as follows. 
Again, we start by guessing the \switchpathtree of an optimal solution. 
This time, since $b$ is a parameter, we can guess in fpt time the allocated budget for delaying and advancing at a switch per base-path as well as the interactions of propagation between switches. Then for each base-path, we ``pseudo''-greedily pick the best switch that aligns with these guesses. 
The main difficulty stems from the fact that the switches we pick, both between base-paths and on the same base-path, are {\em not} independent. 
In particular, we have to pick some switches in batches to make sure that the placement of dependent switches displays the guessed interactions, which we achieve via an intricate sequence of arguments.%

A concise overview of our results is depicted in \Cref{tab:results_table}.

\providecolor{CmpP}{RGB}{175,222,175}
\providecolor{CmpNPc}{RGB}{238,182,180}
\providecolor{CmpNone}{RGB}{1,1,1}
\providecolor{CmpHdr}{RGB}{68,68,68}
\providecolor{CmpOpen}{RGB}{255,220,110}

\begin{table}[t]
\footnotesize
\centering
\begin{tblr}{
    width=\textwidth,
    colspec = {@{}Q[l,45]Q[c,50]Q[c,60]Q[c,50]Q[c,50]},
    rowsep = 3pt,
    row{1} = {},
    column{1} = {White},
    cell{2}{2} = {r=2}{CmpP!30},
    cell{2}{3} = {r=2}{CmpNPc!30},
    cell{2}{4} = {c=2}{CmpOpen!30},
    cell{3}{4} = {c=2}{CmpNPc!30},
    cell{4}{2} = {CmpNone!8},
    cell{4}{3} = {c=3}{CmpNPc!30},
    cell{5}{2} = {CmpNone!8},
    cell{5}{3} = {CmpP!30},
    cell{5}{4} = {c=2}{CmpP!30},
    hline{1,6} = {-}{black},
    hline{2} = {-}{black},
    hline{3,4,5} = {-}{black!20},
    vline{2,3,4,5} = {1-5}{black!20},
}
    \textsc{MaxReach}& $b=\infty$ &\textsc{Delay} & \textsc{Advance} & \textsc{Shift} \\
    by $k$
        & \FPT (Thm. \ref{thm:rprm_FPT_by_k})
        & W[1]-h (Thm. \ref{thm:delay_w1_by_k}), \XP (Thm. \ref{thm:budgeted_XP_by_k})
        & \SetCell[c=2]{c} open, \XP (Thm. \ref{thm:budgeted_XP_by_k})
        &
    \\
    by $k$ fxd SPT
        & 
        & 
        & \SetCell[c=2]{c} W[1]-hard (Thm. \ref{thm:fixedSPT_w1_by_k}), \XP (Thm. \ref{thm:budgeted_XP_by_k})
        &
    \\
    by $b$
        & -----
        & \SetCell[c=3]{c} W[2]-h (Thm. \ref{thm:budgeted_w2_by_b}), \XP (Thm. \ref{thm:budgeted_XP_by_b})
        &
        &
    \\
    by $k{+}b$
        & -----
        & \FPT (Thm. \ref{thm:delay_fpt_by_b_and_k})
        & \SetCell[c=2]{c} \FPT (Thm. \ref{thm:all_fpt_by_b_and_k})
        &
        \\
\end{tblr}
    \caption{Overview of our results. All settings are \NP-hard. The rows show the parameterized complexity of the \textsc{MaxReach} problem variants with respect to the number of base-paths $k$, the budget $b$, and their combination; ``fxd SPT'' stands for ``fixed switch-path-tree'' which denotes the special case where the order in which we have to traverse the base-paths is additionally given.
    The $b=\infty$ column represents our results for delaying/advancing/shifting without a budget constraint, so there is no parameterization by $b$.
    }
    \label{tab:results_table}
\end{table}

\subsection{Related work} \label{sec:related work}

\iflong Modifying the temporal edge set of a temporal graph is a frequently considered problem. \fi
We survey work on reachability-motivated modifications which aim to maximize/minimize reachability or satisfy specific travel demands, followed by other temporal graph modification problems and work on the \kpathgraph{k} model.

\bigparagraph{Shifting and scheduling to maximize reachability.}
Two previous works aim to maximize reachability in some sense and are thus closely related to ours.
Deligkas, Eiben and Skretas~\cite{deligkas_minimizingreachability_2023} study the \textsc{ReachFast} problem of minimizing the time for one or multiple sources to reach all vertices in a temporal graph using advancing, delaying, or both (shifting).
\iflong 
They prove \NP-hardness when the budget, \ie the total shift, is bounded and \Wtwo-hardness wrt the budget; for unlimited shifts, they give a polynomial-time algorithm for a single source, while the problem becomes \NP-hard already for two sources.
\fi 
Their setting differs from ours in two key ways: 
    (1) they work on general temporal graphs (not \kpathgraphs{k}), which removes propagating interactions, and 
    (2) they aim to reach all vertices while minimizing travel time, while we aim to maximize the number of vertices reached. %

Brunelli, Crescenzi and Viennot~\cite{brunelli_maximizingreachability_2023} study \textsc{MaxReach Trip Temporalization}: given a collection of walks in a static directed graph, assign a starting time to each walk that turns it into a temporal walk with no waiting at the vertices, such that the number of reachable pairs of nodes in the resulting temporal graph is maximized. They consider one-to-one, one-to-all, and all-to-all reachability variants.
\iflong
Their setting differs from ours in that we are given temporal base-paths whose labels we can shift at any point of the base-path with delay propagation depending on waiting gaps. This makes our limited-budget problem incomparable to theirs, while they coincide in the infinite-budget setting.
\fi%

\bigparagraph{Edge modification \iflong to restrict reachability. \else for other objectives.\fi}
\iflong Another line of work aims to \emph{restrict reachability} under different types of label manipulations. The corresponding problems are \textsc{(Min)MaxReach} and \textsc{(Max)MinReach}: minimizing the (maximum) or maximizing the (minimum) reachability set size from a single or multiple sources.
Enright, Meeks and Skerman~\cite{enright_assigningtimes_2021} show \NP-hardness of \textsc{MinMaxReach} with one source under reordering of edge sets. %
Enright, Meeks, Mertzios and Zamaraev~\cite{enright_deletingedges_2021} study \textsc{MinMaxReach} with one source under edge deletion, proving \NP-hardness but also giving \FPT algorithms. %
Deligkas and Potapov~\cite{deligkas_optimizingreachability_2022} study all combinations of \textsc{(Max)MinReach} and \textsc{(Min)MaxReach} under edge delay or edge merge — the latter batching consecutive edges to the latest label — and provide a thorough complexity investigation.
Molter, Renken and Zschosche~\cite{molter_temporalreachability_2024} study \textsc{MinReach} with multiple sources under edge deletion or edge delay, establishing \Wone-hardness for deletion but \FPT for delay parameterized by the number of reachable vertices.\else
There are many previous works on restricting the reachability of one or multiple sources via edge modifications~\cite{deligkas_optimizingreachability_2022,enright_deletingedges_2021,enright_assigningtimes_2021,molter_temporalreachability_2024}. 
\fi
Enright, Larios-Jones, Meeks and Pettersson~\cite{enright_reachabilitytemporal_2025} study reachability under uncertainty, where $\zeta$ edge labels may be perturbed by $\pm\delta$\iflong; the problem is intractable in general but efficiently solvable when $\zeta$ is large\fi.
Kutner and Larios-Jones~\cite{kutner_temporalreachability_2023} introduce the \emph{Temporal Reachability Dominating Set} (TaRDiS) problem: can a population be fully reached from a small set of $k$ initially infected individuals, and give parameterized complexity results. %
\ifshort
Kutner and Sommer~\cite{kutner_betterlate_2025} study the \textsc{DelayBetter} problem where a fixed set of passengers’ travel demands must be satisfied, rather than a global reachability objective.
\fi
\iflong

\bigparagraph{Satisfying travel demands.}
Rather than optimizing a global reachability objective, a separate line of work asks whether edge delays can be chosen to simultaneously satisfy a fixed set of passengers’ travel demands.
Kutner and Sommer~\cite{kutner_betterlate_2025} study this as the \textsc{DelayBetter} problem: given a collection of source, destination, and desired arrival time of individual passengers, can some edges be delayed to realize all those demands? 
They give polynomial-time algorithms for variants where passengers also specify their route, and an \FPT algorithm parameterized by the feedback edge set size and the number of passengers, while the general problem is \NP-complete even for small delays.
More broadly, \emph{delay management} in public transport networks asks for a ``good'' delaying strategy to minimize (1) passenger inconvenience, which is typically measured by total passenger delay, (2) the number of delayed trains or (3) operational costs.
An overview of this area %
is given in \cite[Section~1.2]{kutner_betterlate_2025}. 
\fi

\bigparagraph{$k$-Path Graphs.}
The \kpathgraph{k} model was introduced by Deligkas, Döring, Eiben, Goldsmith, Skretas, Tennigkeit~\cite{deligkas_howmany_2025}, who define the \textsc{Exact Edge-Cover} problem as a mathematical way of identifying temporal graphs as public transport networks. An exact edge-cover is a decomposition of the temporal edge set into temporal-edge-disjoint trips. This naturally leads to the definition of a \kpathgraph{k} as a temporal graph being given as a union of temporal paths, as well as the \emph{temporal path number} as the minimum number of temporal paths that a temporal graph's edge set can be partitioned into.
The same authors study the parameterized complexity of temporal connected components %
with respect to the temporal path number in \cite{deligkas_parameterizedcomplexity_2025}.

\section{Preliminaries}\label{sec:prelims}
    A \textit{temporal graph} $\gcal=(V,E,\lambda)$ consists of a static graph $G=(V,E)$, called the \textit{footprint}, along with a labeling function $\lambda$.
    All temporal graphs we consider in this paper are directed.
    A pair $(e, t)$, where $e \in E$ and $t \in \lambda(e)$, is a \textit{temporal edge} with \textit{label} $t$. We denote the set of all temporal edges by $\ecal$.
    The range of $\lambda$ is referred to as the \textit{lifetime} \lifetime. The static graph $G_t = (V, E_t)$, where $E_t = \{e \in E \colon t \in \lambda(e)\}$, is called the \textit{snapshot} at time $t$.
    
    A \textit{temporal path} is a sequence of temporal edges $\tuple{(e_i,t_i)}$ where $\tuple{e_i}$ forms a path in the footprint and the time labels $\tuple{t_i}$ are strictly increasing. %
    Given a temporal path $P$ and vertices $u, v$ on $P$, we write \subpath{P}{u}{v} for the temporal sub-path of $P$ starting from $u$ and ending at $v$. We further write \suffix{P}{u} for the suffix starting at $u$, and \prefix{P}{v} for the prefix ending at $v$.
    Because we never consider static paths and only use temporal graphs in this paper, we will often refer to them simply as \emph{paths} for brevity. If there exists a temporal path from $u$ to $v$, we say \emph{$u$ reaches $v$}. We refer to the set of vertices that $v$ reaches in \gcal as $R^\gcal(v)$.

\bigparagraph{Parameterized complexity.}
    We refer to the standard books for a basic overview of parameterized complexity theory~\cite{cygan_parameterizedalgorithms_2015,downey_fundamentalsparameterized_2013,fomin_kernelizationtheory_2019}.
    A problem is \emph{fixed-parameter tractable} (\FPT) by a parameter $k$ if it can be solved in time $f(k) \cdot \poly(n)$, where $f$ is a computable function.
    Showing that a problem is $\Wone$-hard parameterized by $k$ rules out the existence of such an \FPT algorithm under the assumption $\Wone \neq \FPT$ (the same goes for $\Wtwo$-hardness).
    Problems admitting an algorithm with an exponential running time $\bigoh(n^{f(k)})$ (for some computable function~$f$) belong to the class~$\XP$.

\section{\kPathGraphs{k} and Path-Shifting}\label{sec:properties_of_kpathsgraphs}

We consider the class of temporal graphs that are represented as a union of $k$ temporal paths. In order to avoid any confusion between the $k$ temporal paths that define the graph and other temporal paths in the graph, we will refer to the former as {\em base-paths}.
    \begin{definition}[\kPathGraphs{k}]
        A \emph{\kpathgraph{k}} $\gcal=(V,E, \lambda)$ is a temporal multigraph\footnote{We define \kpathgraphs{k} as multigraphs because shifting operations might cause an edge to have the same label as another edge between the same pair of vertices. Without multigraphs, such an edge would effectively be removed from the graph and it might no longer be a \kpathgraph{k}.} for which there exists a collection $\pcal=\{P_1,\dots,P_k\}$ of $k$ temporal base-paths such that the multiset union of their edges is exactly \ecal.
        We may denote such a graph as $\gcal=\biguplus\pcal=\biguplus_{i\in[k]}P_i$.
    \end{definition}

    For vertices $v,w$ on the same base-path $P$ we say that $v<_P w$ iff $v$ occurs on $P$ before $w$.%

    \bigparagraph{Base-Path Shifting.}
    A \emph{shifting operation} on a \kpathgraph{k} $\gcal$ is a pair $((e, t), \delta)$ where $(e, t)\in \ecal$ and $\delta\in \Int$. Applying $((e, t), \delta)$ to $\gcal$ results in a \kpathgraph{k} $\gcal'$ where $(e, t)$ is replaced with $(e, t+\delta)$. If $\delta > 0$, we call it a \emph{delaying operation}, and the delay propagates to subsequent edges on the unique base-path $P_i$ containing $(e,t)$: if $e'$ is the $d$-th edge after $e$ on $P_i$, its label is increased to $t+\delta+d$ if it is currently lower.
    Likewise if $\delta < 0$, we call it an \emph{advancing operation}, and it propagates to the preceding edges on the same base-path: if $e'$ is the $d$-th edge before $e$ on $P_i$, its label is reduced to $t+\delta-d$ if it is currently higher. 
    Formally, applying the shifting operation $((e, t), \delta)$ with $(e, t)\in P_i$ to $\gcal$ creates a new \kpathgraph{k} $\gcal'=(V, E, \lambda')$ with the labeling function $\lambda'$ defined as follows:
    \begin{align}
        \forall e' \in E \colon \lambda'(e')= \begin{cases}
            \lambda(e'), &\text{if } e'\notin P_i,\\
            t+\delta, &\text{if } e'=e,\\
            \max(\lambda(e'), t+\delta+d), &\text{if $e'$ is the $d$-th edge after $e$ on $P_i$},\\
            \min(\lambda(e'), t+\delta-d), &\text{if $e'$ is the $d$-th edge before $e$ on $P_i$}.
        \end{cases}
    \end{align}
    The \emph{cost} of an operation $((e, t), \delta)$ is $\lvert\delta\rvert$. Propagation is not accounted in the cost.
    A sequence of shifting operations \scal is applied one after another, and its cost is the sum of costs of all operation. See \Cref{fig:shifting_operation_example} for an example.

    \begin{figure}[t]
        \centering
        \includegraphics[width=0.7\linewidth]{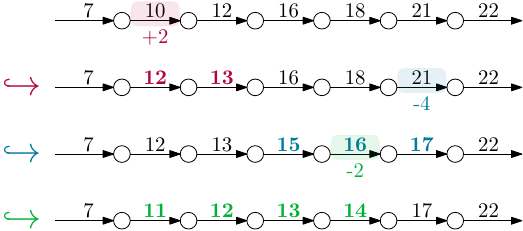}
        \caption{An example sequence of three shifting operations applied to a base-path.}
        \label{fig:shifting_operation_example}
    \end{figure}

    Note that the order of shifting operations matters (for example, switching the order of the last two operations in \Cref{fig:shifting_operation_example} would return a different labeling). However, it is generally optimal to apply delaying operations beginning-to-end along a base-path, and apply advancing operations in the reverse order; this way, operations are applied on top of the propagation of prior operations, maximizing the overall change caused.
    We further observe that there it is never efficient to cause a delay that propagates into previously advanced edges as this would undo part of the previous advancing operation, and vice versa.

    \iflong
    \begin{lemma}\label{cor:delay_and_advance_should_never_collide}
        If a sequence of shifting operations causes a previously advanced edge to be delayed (directly or through propagation) or vice versa, then this sequence does not have minimal cost.
    \end{lemma}
    \begin{proof}
        Let $\scal$ be such a sequence of operations: Let $e\in P$ be an edge of base-path $P$ that is wlog advanced due to the operation $(e_a, -a)$ (where $e_a=e$ or $e <_P e_a$) and then delayed due to the operation $(e_d, +d)$ (where $e_d=e$ or $e_d <_P e$). Replacing the operation $(e_a, -a)$ in \scal with $(e_a, -a+1)$ reduces the overall cost of \scal but yields the same labeling: Because $e$ (and as well as the following edges until $e_a$ if $e <_P e_a$) are delayed by the delaying operation, at least one unit of advance from the advancing operation was negated anyway.
    \end{proof}
    \fi
    With these definitions and insights in place, let us restate the problem definition:

\begin{problem}[]{{ \SPRM (\sprm)}}
    \Input &A \kpathgraph{k} $\gcal$, a budget $b\in \Na$, and a vertex $s$.\\
    \Prob & {Find a sequence \scal of shifting operations with a total cost of at most $b$, such that $|R^{\gcal'}(s)|$ is maximized, where $\gcal'$ is the resulting \kpathgraph{k} after applying \scal to \gcal.}
\end{problem}

If we allow only delaying or only advancing operations, we call the problem \DPRM (\dprm) or \APRM (\aprm) respectively.

\begin{figure}[t]
    \centering
    \includegraphics[width=0.7\linewidth]{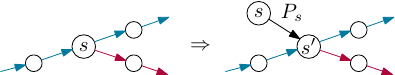}
    \caption{A \kpathgraph{k} (left) without a distinct base-path $P_s$ and its adjustment \kpathgraph{k+1} (right)  with an additional base-path $P_s$ (black) of length 1.}
    \label{fig:source_path}
\end{figure}
To simplify notation later, we slightly adjust our representation of the given source $s$. In the rest of this paper, we will assume that there is exactly one base-path that contains $s$, which we will call $P_s$. For \kpathgraphs{k} where this is not the case, we rename $s$ to $s'$, add a new source vertex $s$, and a new base-path $P_s$ with only the edge $(s, s')$ that has a very small label (smaller than all other labels in the graph, even if they get advanced by a reasonable amount). See \Cref{fig:source_path}.
This does not affect the studied problem as any path that started in $s$ before the adjustment is now simply preceded by this new edge.

\bigparagraph{Properties of \kPathGraphs{k}.} %
    We begin with a fundamental observation about \kpathgraphs{k} that we will use throughout the paper:

\begin{lemma}\label{lem:only_one_segment_of_each_basepath_needed}
        In a \kpathgraph{k}, for any two vertices $u, v$ where $u$ reaches $v$, there is a temporal path from $u$ to $v$ that uses at most one contiguous segment from each base-path.
    \end{lemma}
    \iflong
    \begin{proof}
        If $u$ reaches $v$, there is a path $P$ from $u$ to $v$. If $P$ uses at most one contiguous segment from each base-path, the claim follows. Otherwise, $P$ contains two separate segments $\subpath{P_i}{a}{b}$ and $\subpath{P_i}{c}{d}$ of the same base-path $P_i$. We can then replace the subpath $\subpath{P}{a}{d}$ with one contiguous segment $\subpath{P_i}{a}{d}$, thus decreasing the number of segments. As we can repeat this as long as $P$ contains at least two separate segments from some base-path, this will terminate when $P$ contains at most one segment from each base-path.
    \end{proof}
    \fi
    Furthermore, if there are two paths starting from a source $s$ that switch onto base-path $P_i$ at different vertices, both paths could switch onto $P_i$ using the vertex that is earlier on $P_i$ (see \Cref{fig:switch_point_intuition}).
    Crucially, this means that we can maintain the entire reachability of $s$ even when we restrict paths to only switch onto a base-path at one designated switch-vertex.
    
\begin{figure}[t]
    \centering
    \includegraphics[width=0.4\linewidth]{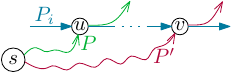}
    \caption{
    The two paths $P$ and $P'$ both start at $s$ and use a segment of the base-path $P_i\in \pcal$, but $P$ switches onto it at $u$ and $P'$ does so at $v$. Because $u$ is before $v$ along $P_i$ we could replace \prefix{P'}{v} in $P'$ with $\prefix{P}{u}\circ \subpath{P_i}{u}{v}$, making both switch onto $P_i$ in the same vertex.}
    \label{fig:switch_point_intuition}
\end{figure}

\bigparagraph{Switch-Vertex-Sets and Switch-Path-Trees.} 
Based on the previous insight, we can represent the vertices reachable from the source using solely the base-paths and switches between them. Since every path in \gcal starting from $s$ starts with the base-path $P_s$ and then continues onto other base-paths, this yields a tree structure rooted in $P_s$. In the following, we define the mathematical objects that model this representation, beginning with a set of vertices that uniquely characterizes the points at which any path starting in $s$ switches between base-paths. See \Cref{fig:switch_vertex_set_example} for an example.
\begin{definition}[Switch-vertex-set]\label{def:switchverts}
    For a \kpathgraph{k} $\gcal = \biguplus\pcal$ with vertex set $V$, let \svset be a subset of $V \times \pcal \times \pcal$. \svset is a \emph{switch-vertex-set} (\switchvertexset) if the following holds: \begin{itemize}
        \item \textbf{Switch exists:} For each $(v, P_{from}, P_{to})\in \svset$, $P_{from}$ contains an edge going into $v$, and $P_{to}$ contains an edge going out of $v$.
        \item \textbf{At most one switch to each base-path:} For each $P\in \pcal\setminus \{P_s\}$, there is at most one $(v, P_{from}, P_{to})\in \svset$ with $P=P_{to}$. We refer to this $v$ as $\switchonto{P}$. There is no switch onto $P_s$.
        \item \textbf{Switch onto base-path before switching off:} For each $(v, P_{a}, P_{b})\in \svset$, all $v'$ with $(v', P_b, P_c)\in \svset$ are on $P_b$ after $v$. 
        \item \textbf{Base-path switches form a tree:} Let $\sptree_\svset$ be a directed graph on $\pcal$ defined via the edge-set $\{(P_{from}, P_{to}): \exists (v, P_{from}, P_{to})\in \svset\}$. $\sptree_\svset$ is a directed tree rooted at $P_s$.
    \end{itemize}
\end{definition}

Note that this definition disregards time labels. We will separately distinguish whether a specified switch or switch-vertex-set is temporal. 

\begin{definition}[temporal switch]\label{def:temporalswitch}
    We say that a switch $(v, P_{from}, P_{to})$ is \emph{temporal in $\gcal$} if the edge of $P_{from}$ going into $v$ has a smaller label than the edge of $P_{to}$ going out of $v$. An \switchvertexset \svset is \emph{temporal in \gcal} if all contained switches are temporal in \gcal.
\end{definition}

We are interested in paths that only switch onto base-paths according to a given \switchvertexset, as per the following definition: 

\begin{definition}[\svset-respecting, \svset-reachability]\label{def:svset-respecting}
    We say that a path in $\gcal$ is \emph{\svset-respecting} if, for any two consecutively edges $((u, v),t_1)\in P_a$ and $((v, w), t_2)\in P_b\neq P_a$, we have $(v, P_a, P_b)\in \svset$.
    Furthermore, the \emph{\svset-reachability} $R^\gcal_\svset(s)$ of a vertex $s$ is the set of all vertices $v$ for which there is a \svset-respecting path from $s$ to $v$.
\end{definition}

\begin{figure}[t]
    \centering
    \includegraphics[width=1\linewidth]{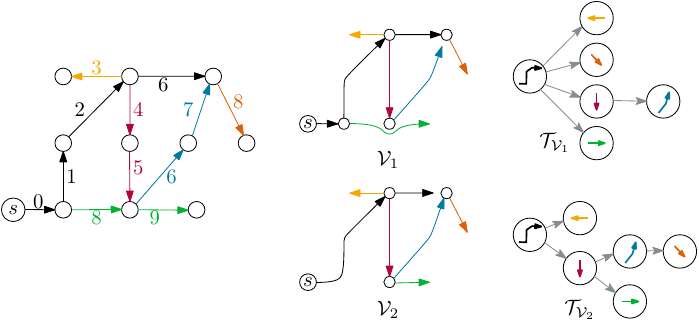}
    \caption{A \kpathgraph{6} \gcal and two \switchvertexsets $\svset_1, \svset_2$ (see \Cref{def:switchverts}) that are both temporal in \gcal, along with their respective implied $\switchpathtrees$ (see \Cref{def:switchpathtree}).}
    \label{fig:switch_vertex_set_example}
\end{figure}

If an \switchvertexset \svset is temporal, the \svset-reachability is simply the combined suffixes of each base-path starting from their respective switch:

\begin{lemma}\label{lem:svset_respecting_reachability}
    Let \gcal be a \kpathgraph{k} with a source $s$ and let \svset be an \switchvertexset that is temporal in \gcal. Then $R^\gcal_\svset(s)=\bigcup_{P\in \pcal}\{v\in \suffix{P}{\switchonto{P}}\}$.
\end{lemma}
\iflong
\begin{proof}
    We show that every vertex reachable by $s$ is on the suffix of some base-path and vice versa. First, if a vertex $v$ is in $R^\gcal_\svset(s)$, then there is a $\svset$-respecting path from $s$ to $v$; let the last edge on this path belong to base-path $P$. Then $v$ is on $\suffix{P}{\switchonto{P}}$ as the path must have switched onto $P$ at $\switchonto{P}$ earlier.

    Second, if a vertex $v$ is on some $\suffix{P}{\switchonto{P}}$ for $P \in \pcal$, then we can derive a $\svset$-respecting path from $\sptree_\svset$: Because $\sptree_\svset$ forms a tree rooted at $P_s$, it contains a sequence of base-paths $P_s, P_1, P_2, \dots, P$. We can construct a path in $\gcal$ from $s$ to $\switchonto{P_1}$, and then from any $\switchonto{P_i}$ to $\switchonto{P_{i+1}}$ by following $P_i$. Because all switches are temporal, this together forms a temporal path, and thus $v\in R^\gcal_\svset(s)$.
\end{proof}
\fi
Next, we show that for maximizing the reachability of $s$, it is sufficient to only consider reachability with respect to a specific \switchvertexset.

\begin{lemma}\label{lem:switchverts_is_enough}
    For a \kpathgraph{k} $\gcal$ and a vertex $s$, there exists a temporal \switchvertexset \svset such that $R^\gcal(s)=R^\gcal_\svset(s)$.
\end{lemma}
\iflong
\begin{proof}
    Let $\gcal=\biguplus\pcal$ be a \kpathgraph{k}, and let $s$ be a vertex in \gcal. Construct an \switchvertexset \svset for $s$ as follows:
    
    For every $P_i \in \pcal \setminus \{P_s\}$, determine the first edge $e_i=((v, w), t_1)$ of $P_i$ that could be traversed in a path starting from $s$.
    On that path, let $e_j=((u, v), t_2)$ be the edge directly preceding $e_i$. By choice of $e_i$, $e_j$ is part of a different base-path $P_j$. Add $(v, P_j, P_i)$ to $\svset$.

    We now prove that $R^\gcal(s)=R^\gcal_\svset(s)$. By definition, any vertex in $R^\gcal_\svset(s)$ is also in $R^\gcal(s)$. On the other hand, if a vertex is in $R^\gcal(s)$, then there is a path from $s$ to $v$ where the last edge is part of a base-path $P$. Then \svset contains some vertex on $P$ that is before this edge. By \Cref{lem:svset_respecting_reachability}, $v$ is then also in $R^\gcal_\svset(s)$.
\end{proof}
\fi
An \switchvertexset further simplifies the shifting operations we might consider to optimize reachability. Given an \switchvertexset \svset as per the previous lemma, we find that reachability can be maximized by only applying operations to the edges referenced in some switch of \svset.

\begin{lemma}\label{lem:canonical_shift_form}
    Given a $k$-path graph $\gcal$, there is an optimal solution $\gcal^*$ and a \switchvertexset $\svset$ with $R^{\gcal^*}(s)=R^{\gcal^*}_\svset(s)$ such that $\gcal^*$ can be constructed via a sequence of shifting operations where for each $(v_P^*, \swPfrom, P)\in \svset$ there is: %
    \begin{itemize}
        \item at most one advance $\alpha_{P} \le 0$ applied to the temporal edge of $\swPfrom$ entering $v_P^*$, and
        \item at most one delay $\delta_{P} \ge 0$ applied to the temporal edge of $P$ leaving $v_P^*$.
    \end{itemize}
    If only one of the two operation types is allowed, then we omit the other from the statement.
\end{lemma}
\iflong
\begin{proof}
    Consider an optimal solution $\gcal'$. We assume by \Cref{cor:delay_and_advance_should_never_collide} that in the sequence of operations creating $\gcal'$, advances and delays never collide. By Lemma~\ref{lem:switchverts_is_enough}, there is a \switchvertexset $\svset$ such that $R^{\gcal'}(s)=R^{\gcal'}_\svset(s)$. If $\gcal'$ and $\svset$ already satisfy the lemma, we are done. Otherwise, we show that we can move the shifting operations to just the edges entering and leaving a switch without compromising the \svset-reachability (which is equal to the overall reachability by \Cref{lem:switchverts_is_enough}).

    Let $((e, t), \delta)$ with $\delta > 0$ be a delaying operation applied to an edge of a base-path $P$ such that $e$ is not leaving $v_P^*$.
    We can instead apply a delay on the subsequent edge $e'$ along $P$, but now only delay by the amount that would have propagated to $e'$ from the original operation (if $e$ is the last edge of $P$, instead remove the operation). The resulting graph only differs in that $e$ has a smaller label. If a switch that was temporal before becomes non-temporal from this, then $e$ must be the edge going out of the switch, but because $e$ is on $P$ and not leaving $v_P^*$, this is not possible.

    Let $((e, t), \alpha)$ with $\alpha < 0$ be an advancing operation applied to an edge of base-path $P$ such that $e$ is not entering $v_{P_{to}}^*$ for any $(v_{P_{to}}^*, P, P_{to})\in \svset$. We can instead apply an advance on the preceding edge $e'$ along $P$, but now only advance by the amount that would have propagated to $e'$ from the original operation (if $e$ is the first edge of $P$, instead remove the operation). The resulting graph only differs in that $e$ has a higher label. If a switch that was temporal before becomes non-temporal from this, then $e$ must be the edge entering the switch, but because $e$ is on $P$ and not entering any of the switches away from $P$, this is not possible.

    Note that no exchange can increase the cost of the operations. We repeat these exchanges until all delays are applied to edges entering a switch, and all advances are applied to edges exiting a switch. If multiple delays or advances are then applied to the same switch, combine them into one operation applying the summed up amount.
\end{proof}
\fi

Lastly, we define \emph{switch-path-trees} which are a simpler representation as they only specify which base-path switches onto each base-path but not the vertices used to switch. In particular, every switch-vertex-set \svset defines a switch-path-tree $\sptree_\svset$, though not every switch-path-tree is defined by some switch-vertex-set. See \Cref{fig:switch_vertex_set_example} for examples.

\begin{definition}[switch-path tree]\label{def:switchpathtree}
    For a $\gcal = \biguplus\pcal$, let $\sptree=(V_\sptree, E_\sptree)$ be a directed graph with $V\subseteq\pcal$ (meaning every vertex in \sptree represents a base-path). \sptree is a \emph{switch-path-tree} (or \emph{\switchpathtree}) if it %
        is a directed tree rooted at $P_s$.

\end{definition}

    Akin to \Cref{def:svset-respecting}, a path in $\gcal$ is \sptree-respecting if, for any two consecutively edges $((u, v),t_1)\in P_a$ and $((v, w), t_2)\in P_b\neq P_a$, we have $(P_a, P_b)\in E_\sptree$. The \sptree-reachability $R^\gcal_\sptree(s)$ of a vertex $s$ is the set of all vertices $v$ for which there is a \sptree-respecting path from $s$ to $v$.

\ifshort
\medskip
We conclude by providing parameterized algorithms for enumerating all \switchvertexsets and \switchpathtrees.
\fi

\iflong
Now we analyze the time it takes to enumerate all different \switchvertexsets and \switchpathtrees of a \kpathgraph{k}. 
 Since there are at most $k=|\pcal|$ switches from one base-path to another and each uses one of $|V|$ vertices, we can enumerate all switch-vertex-sets in \XP time parameterized by $k$.
 \fi
\begin{lemma}\label{lem:svs_enumeration}
    All \switchvertexsets of a \kpathgraph{k} with $n$ vertices can be enumerated in $\bigoh\!\left((kn)^{k-1}\right)$ time.
\end{lemma}
\iflong
\begin{proof}
    An \switchvertexset has at most one switch onto each non-source base-path, \ie at most $k-1$ entries in total.
    Each $(v, P_j, P_i)$ is determined by choosing the predecessor base-path $P_j \in \pcal$ (at most $k$ choices) and the switch vertex $v \in V(\gcal)$ (at most $n$ choices). Enumerating all combinations of these independent choices takes $\bigoh\!\left((kn)^{k-1}\right)$ time.
\end{proof}
\fi
\iflong
In contrast, switch-path-trees can be enumerated in \FPT time parameterized by $k$, since we only have to enumerate over the combinations of base-path switches.
\fi
\begin{lemma}\label{lem:spt_enumeration}
    All \switchpathtrees of a \kpathgraph{k} can be enumerated in $\bigoh\!\left(k^{k-2}\right)$ time.
\end{lemma}
\iflong
\begin{proof}
    An \switchpathtree is a directed rooted tree on $k$ labeled vertices (the base-paths), rooted at $P_s$.
    This is equivalent to choosing a labeled rooted forest on the $k-1$ non-source base-paths and connecting each forest root to $P_s$. By Cayley's formula \cite{cayley_theoremtrees_1897}, the number of such forests, and hence the number of \switchpathtrees, is $k^{k-2}$, which depends only on $k$.
\end{proof}
\fi

\section{Unconstrained Budget: Path Temporalization}\label{sec:infinite_budget}

To get familiar with the problem, we will first consider a restricted setting without a budget constraint, \ie the budget is infinite. %

For this setting we observe that the permitted type of operations and even the original base-path labels are irrelevant: 
Without a limited budget, we can delay or advance each edge on each base-path by an arbitrary amount. The propagation of these shifts will keep the labels on each path strictly increasing, thereby retaining the temporal path property. This way, we can create any relative temporal order of the edges, so long as edges along a base-path are strictly increasing.
Because temporal reachability only requires the traversed edges to have increasing labels, only the relative temporal order of edges is relevant for our problem. 
Thus neither the initial labels nor the type of allowed operation have an effect. To simplify, we define the following equivalent \emph{temporalization problem}:

\begin{problem}[]{\RPRM (\rprm)}
    \Input & A collection of static paths $\mathbb{P}$ and a vertex $s$.
    \\
    \Prob & {
    Find a labeling $\lambda$ for the edges of each path in $\mathbb{P}$ such that the labels along each paths are strictly increasing and the number of vertices that $s$ can reach in the resulting \kpathgraph{k} $(\biguplus\mathbb{P},\lambda)$ is maximized.}
\end{problem}

This problem is closely related to the \textsc{TripTemporalization} problem introduced by~\cite{brunelli_maximizingreachability_2023}. Here, we are given a collection of walks and have to assign a starting time to each walk. This starting time becomes the label of the first edge of the walk, and all following edges are labeled incrementally\footnote{More precisely, each following edge is labeled with the accumulated travel time of all preceding edges. However, the authors show that it suffices to consider instances in which every edge has travel time~1.}.
In this way, the collection of walks is turned into a temporal graph, and the goal is to maximize some notion of reachability in that resulting graph.

The two problems are related in the following way: On the one hand, \cite{brunelli_maximizingreachability_2023} allow walks instead of paths and are thus slightly more general. On the other hand, their model is also slightly more restrictive, since the labels along a walk are determined entirely by its starting time, and no additional waiting times can be inserted.

Since their input consists of a collection of walks, we will refer to this operation as \emph{walk temporalization} (instead of trip temporalization). The specific variant in which the goal is to maximize the reachability of a given source is then called \textsc{MaxReach-WalkTemporalization} (\textsc{MR-WT}), and is proven to be \NP-hard. We show that this hardness transfers to \rprm.
\begin{theorem}\label{thm:rprm_NP-hard}
    \rprm is \NP-hard.
\end{theorem}
\iflong
\begin{proof}
    We first note that in the hardness construction provided by \cite{brunelli_maximizingreachability_2023} all the walks are paths, so the problem of assigning starting times to static paths is consequently also hard.
    In the following, we show that \textsc{MR-WT} and \rprm are equivalent.

    Naturally, any labeled graph $\gcal'$ that can be created in \textsc{MR-WT} can also be created in \rprm, since \rprm permits strictly more labelings as solutions than \textsc{MR-WT}. We now show that, for any labeled graph $\gcal^{PT}$ created in \rprm, there is a labeled graph $\gcal^{WT}$ that can be created in \textsc{MR-WT} such that $R^{\gcal^{PT}}(s)=R^{\gcal^{WT}}(s)$.
    By \Cref{lem:switchverts_is_enough}, there is a temporal \switchvertexset \svset with $R_\svset^{\gcal^{PT}}(s)=R^{\gcal^{PT}}(s)$. Let \sptree be the \switchpathtree implied by \svset. For \textsc{MR-WT}, let $\gcal^{WT}$ be constructed such that each base-path has the starting time $d\cdot M$, where $d$ is the BFS-depth of $P$ in \sptree. This way, any switch from a base-path $P$ to each other base-path it has an edge to in \sptree is temporal. Thus \svset is a temporal \switchvertexset in $\gcal^{WT}$, and $s$ has the same reachability in $\gcal^{WT}$ as in $\gcal^{PT}$.
\end{proof}
\fi

Despite this hardness result, we can show that there exists an \FPT algorithm parameterized by~$k$. For this, we first show that given a switch-path-tree \sptree, there is an algorithm to determine in polynomial time a switch-vertex-set \svset from among all switch-vertex-sets that imply \sptree, such that the reachability of $s$ with respect to $\svset$ is maximized.

\begin{lemma}\label{lem:algo_for_best_SVS_to_a_given_SPT}
    Given a \kpathgraph{k} \gcal, a source $s$, and an \switchpathtree \sptree, one can determine in polynomial time an \switchvertexset \svset with $\sptree_\svset=\sptree$ such that for all other \switchvertexsets $\svset'$ with $\sptree_{\svset'}=\sptree$ we have $|R^\gcal_{\svset'}(s)|\leq |R^\gcal_\svset(s)|$, if any such \switchvertexsets exist.
\end{lemma}
\iflong
\begin{proof}
    For each edge $(P_{from}, P_{to})$ in \sptree (representing a switch from $P_{from}$ to $P_{to}$) we need to determine a vertex of \gcal to perform this switch on.

    During this algorithm, we mark some edges of \gcal as \emph{reached}. Initially the only edges reached are those on the base-path $\suffix{P_s}{s}$. Now we traverse \sptree via a BFS starting from $P_s$. For each traversed edge $(P_{from}, P_{to})$ in \sptree, we choose the first vertex $v$ in \gcal on $P_{to}$ that has an incoming reached edge from $P_{from}$ as the switch-vertex, and mark all edges of $\suffix{P_{to}}{v}$ in \gcal as reached. If there is no such vertex, then there is no \switchvertexset that implies \sptree. After the BFS on \sptree, we return the \switchvertexset made up of all chosen switch-vertices.

    First, if the algorithm returns an \switchvertexset \svset, then it must be valid because (i) we only add valid switches, (ii) we add only one switch per edge in \sptree, and (iii) due to only adding switch-vertices with an incoming reached edge, we only add a switch from a base-path $P$ that is further back on $P$ than the switch onto $P$. %

    We show optimality via an exchange argument. Let $\svset'$ be a different \switchvertexset that implies~\sptree. Let $(P_{from}, P_{to})$ be an edge in $\sptree$ with minimal BFS-depth where $\svset$ chose $v$ as the switch-vertex and $\svset'$ instead chose $v'\neq v$. Due to our choice, all switches on the way to $P_{from}$ are identical in $\svset$ and $\svset'$. Now we can exchange $v'$ with $v$ in $\svset'$ without reducing the reachability of $s$: $v$ is earlier on $P_{to}$ than $v'$ as per the algorithm, so any path that used the switch $(v', P_{from}, P_{to})$ can be replaced with a path that instead uses the switch $(v,P_{from},P_{to})$ and then follows $P_{to}$ until $v'$. Thus the reachability of $s$ does not decrease after any number of such exchanges until $\svset$ and $\svset'$ are equal, implying $|R^\gcal_{\svset'}(s)|\leq |R^\gcal_{\svset}(s)|$.
\end{proof}
\fi

With the previous lemma, we can efficiently find an optimal switch-vertex-set for a given switch-path-tree. We can combine this with the fact that there are only $f(k)$ many switch-path-trees\ifshort , \fi \iflong(\Cref{lem:spt_enumeration}),\fi and then describe how to make any switch-vertex-set temporal using infinite budget. This yields the following \FPT result.

\begin{theorem}\label{thm:rprm_FPT_by_k}
    \rprm is \FPT when parameterized by $k$.
\end{theorem}
\iflong
\begin{proof}
    For a given \kpathgraph{k} \gcal and a source $s$, there is an optimal solution $\gcal^*$, \ie a relabeled graph such that $|R^{\gcal^*}(s)|$ is maximized. By \Cref{lem:switchverts_is_enough}, there exists an \switchvertexset $\svset^*$ for~$\gcal^*$ witnessing this reachability, meaning $R_{\svset^*}^{G^*}(s)=R^{G^*}(s)$.
    We can guess its implied \switchpathtree $\sptree_{\svset^*}$ in $f(k)$ time as per \Cref{lem:spt_enumeration}. Using \Cref{lem:algo_for_best_SVS_to_a_given_SPT}, we can then find an \switchvertexset $\svset$ with $\sptree_{\svset}=\sptree_{\svset^*}$ that maximizes the reachability of $s$. So by choice of $\svset^*$, $|R_{\svset}^{\gcal^*}(s)|$ is at least $|R_{\svset^*}^{\gcal^*}(s)|=|R^{\gcal^*}(s)|$, meaning a labeling is optimal if $\svset$ is temporal in it.

    To construct a labeling where \svset is temporal, label the first edge of each base-path $P$ with~$d\cdot M$, where $d$ is the BFS-depth of $P$ in $\sptree_\svset$, and label all consecutive edges incrementally. This makes each switch in \svset temporal, therefore \svset is temporal and the reachability of $s$ is maximized as described above.
\end{proof}
\fi

\section{Limited Budget}\label{sec:limited_budget}

We now turn to the originally defined problem where we have to adhere to a limited budget $b$. This means that a solution must indeed be sufficiently similar to the given labeling and we cannot construct a labeling from scratch. We show additional hardness results for this setting, followed by an \XP algorithm by $k$ and \FPT algorithms by $b+k$.

\subsection{Negative Results}\label{sec:negative_results}

It is already known that this version of the problem is \NP-hard, as well as \Wtwo-hard by $b$, for individual edges \cite{deligkas_minimizingreachability_2023}. As a \kpathgraph{k} may consist exclusively of base-paths of length $1$ (\ie individual edges), this hardness directly transfers to our problem.

\begin{theorem}[\cite{deligkas_minimizingreachability_2023}]\label{thm:budgeted_w2_by_b}
    \dasprm is \NP-hard and \Wtwo-hard when parameterized by $b$.
\end{theorem}

\newcommand{\topv}[1]{\ensuremath{#1^\top}}
\newcommand{\botv}[1]{\ensuremath{#1^\bot}}

The natural next question given the previous results is whether we can hope for an \FPT algorithm parameterized by $k$, as existed for the relabeling setting (\Cref{thm:rprm_FPT_by_k}). However, we can show \Wone-hardness via a reduction from \MCIS (\mcis).

For \mcis, we are given a $k\in \Na$ and an undirected static graph $G=(\bigcup_{C \in \ccal}C, E)$, where \ccal contains $k$ disjoint node sets (colors). The question is whether there is a colorful independent set of size $k$, that is, a selection of one node of each color such that no two selected nodes share an edge.

\begin{theorem}\label{thm:delay_w1_by_k}
    \dprm is \Wone-hard when parameterized by $k$.
\end{theorem}

Given an instance $(G, k)$ of \mcis, we construct a \kpathgraph{2k+1} \gcal and a budget $b$ such that a multicolored independent set of size $k$ exists in $G$ if and only if it is possible to make $s$ reach \emph{all} vertices in \gcal using delaying operations within budget $b$. \iflong We begin with the construction and then work towards the proof of equivalence. \fi \ifshort We give an example of the construction below, along with some intuition.\fi

\bigparagraph{Construction.}
We will first describe the color-gadget in which we represent choosing a node as specific delaying operations. Afterwards we show how to represent the edges $E$ of $G$ to restrict the available choices.

For this construction we use time labels that are multiples of some large number $\omega$, sometimes with some additive terms to make the edges along a base-path strictly increasing. We want to choose $\omega$ large enough that the additive terms in our arguments never add up to $\omega$ itself, so that reachability can usually be gleaned from the multiplier of $\omega$ while the additive terms can be disregarded. Choosing $\omega=|V(G)|^4$ is sufficient for this purpose.

\begin{figure}[h]
    \centering
    \includegraphics[width=0.9\linewidth]{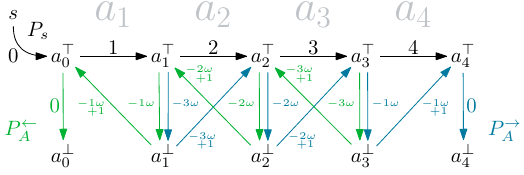}
    \caption{The color-gadget for a color $A\in \ccal$ where $A=\{a_1, a_2, a_3, a_4\}$ nodes. Each node $a_i$ corresponds to the gap between the vertex-pairs $(\topv{a_{i-1}}, \botv{a_{i-1}})$ and $(\topv{a_i}, \botv{a_i})$. In order to reach all $\botv{a}$-vertices, we need to delay $P_A^{\rightarrow}$ (blue) and $P_A^{\leftarrow}$ (green) at adjacent vertex-pairs in order to switch onto them from $P_s$ (black). Such delaying operations cost at least $(|A|-1)\cdot \omega$.}
    \label{fig:delay_w1_by_k_color_gadget}
\end{figure}

We begin with a base-path $P_s$ that starts in a vertex $s$.
Now construct a gadget for each color $A=\{a_1, \dots a_n\}\in \ccal$ as \iflong follows (see \Cref{fig:delay_w1_by_k_color_gadget} for an example):
\begin{enumerate}
    \item For each $0\leq i \leq n$, create a vertex-pair $\topv{a}_i$ and $\botv{a}_i$.
    \item Create a base-path $P_A^\rightarrow$ as $(\topv{a}_1, \botv{a}_1, \topv{a}_2, \botv{a}_2, \dots, \topv{a}_n, \botv{a}_n$).
    \item Create a base-path $P_A^\leftarrow$ as $(\topv{a}_{n-1}, \botv{a}_{n-1}, \topv{a}_{n-2}, \botv{a}_{n-2}, \dots, \topv{a}_0, \botv{a}_0$).
    \item Label each edge $(\topv{a_i}, \botv{a_i})\in P_A^\leftarrow$ with $-i\cdot \omega$, and label each edge $(\topv{a_i}, \botv{a_i})\in P_A^\rightarrow$ with $-(n-i)\cdot \omega$. Other edges on the paths are labeled with their preceding edge's label plus 1.
    \item Extend $P_s$ to all $\topv{a}_i$ with $0 \leq i \leq n$. All edges of $P_s$ are labeled incrementally starting from $0$.
\end{enumerate}
\fi
\ifshort
seen in \Cref{fig:delay_w1_by_k_color_gadget}.
\fi
Intuitively, a node $a_i\in A$ corresponds to the \emph{gap} between the vertex-pairs $\{\topv{a}_{i-1}, \botv{a}_{i-1}\}$ and $\{\topv{a}_{i}, \botv{a}_{i}\}$. Choosing $a_i$ then corresponds to delaying $P_A^\leftarrow$ left of the gap, and delaying $P_A^\rightarrow$ right of the gap.

Now to represent the edges of $G$, we will create additional vertices and insert them into existing base-paths. For each $e=\{a_i, b_j\}\in E$\iflong:
\begin{enumerate}
    \item Create a new vertex $e_{a_i,b_j}$.
    \item If $i>1$, insert $e_{a_i,b_j}$ into $\subpath{ P_A^\rightarrow}{\botv{a}_{i-1}}{\topv{a}_i}$: Replace an arbitrary edge $(u,v,t)$ on $\subpath{ P_A^\rightarrow}{\botv{a}_{i-1}}{\topv{a}_i}$ with $(u, e_{a_i,b_j}, t)$ and $(e_{a_1,b_j}, v, t+1)$, and increase the labels of subsequent edges on the sub-path by $1$.
    \item Likewise, if $i<|A|$, insert $e_{a_i,b_j}$ into $\subpath{P_A^\leftarrow}{\botv{a}_i}{\topv{a}_{i-1}}$.
    \item In the same way, insert it into $P_B^\rightarrow$ and $P_B^\leftarrow$.
\end{enumerate}
\fi
\ifshort
create a new vertex $e_{a_i,b_j}$ and insert it into $P_A^\leftarrow$ and $P_A^\rightarrow$ such that they are in the gap associated with $a_i$. Likewise, insert them into $P_B^\leftarrow$ and $P_B^\rightarrow$.
\fi
See \Cref{fig:delay_w1_by_k_main} for an example.

\begin{figure}[h]
    \centering
    \includegraphics[width=0.8\linewidth]{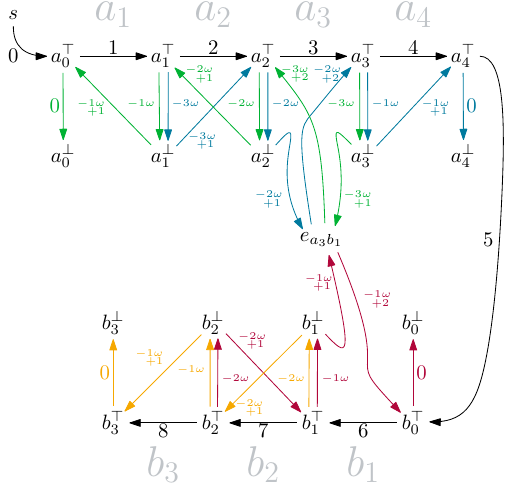}
    \caption{The representation of an edge between $a_3\in A$ and $b_1\in B$ in \gcal.}
    \label{fig:delay_w1_by_k_main}
\end{figure}

We choose the budget of the \dprm instance as $b=\sum_{C\in \ccal}((|C|-1)\cdot \omega)+(\omega-1)$.

\iflong
Now we can show that a multicolored independent set of size $k$ exists in $G$ if and only if it is possible to make $s$ reach \emph{all} vertices in \gcal using delaying operations within the given budget.

\begin{lemma}\label{lem:delay_w1_by_k_IS_implies_dpmr}
    If $G$ admits a multicolored independent set, then the optimal solution to \dprm on $\gcal$ allows $s$ to reach all vertices.
\end{lemma}
\begin{proof}
Let $IS$ be a multicolored independent set of size $k$ in $G$, and for each color $A\in \ccal$, let $IS(A)$ be the index of the node of $A$ that is chosen in $IS$. Based on this, we apply the following delaying operations in \gcal:
\begin{itemize}
    \item Delay $P_A^\rightarrow$ at $\topv{a}_{IS(A)}$ such that switching from $P_S$ to $P_A^\rightarrow$ is possible there.
    
    By construction, the label of $(\topv{a}_{IS(A)}, \botv{a}_{IS(A)})\in P_A^\rightarrow$ is $-(|A|-IS(A))\cdot \omega$, and the label of the edge of $P_s$ going into $\topv{a_{IS(A)}}$ is some $\delta^\rightarrow< |V(G)|$. So the cost of this delay is at most $(|A|-IS(A))\cdot \omega+|V(G)|$.

    \item Delay $P_A^\leftarrow$ at $\topv{a}_{IS(A)-1}$ such that switching from $P_S$ to $P_A^\leftarrow$ is possible there.
    
    By construction, the label of $(\topv{a}_{IS(A)-1}, \botv{a}_{IS(A)-1})\in P_A^\rightarrow$ is $(IS(A)-1)\cdot \omega$, and the label of the edge of $P_s$ going into $\topv{a_{IS(A)-1}}$ is some $\delta^\leftarrow < |V(G)|$. So the cost of this delay is at most $(IS(A)-1)\cdot \omega+|V(G)|$.

\end{itemize}

The cost of both delays for the color $A$ adds up to at most \begin{align*}
    & (|A|-IS(A))\cdot \omega+|V(G)|+(IS(A)-1)\cdot \omega + |V(G)| \\
    = & (|A|-IS(A)+IS(A)-1)\cdot \omega + 2|V(G)| \\
    = & (|A|-1)\cdot \omega + 2|V(G)|.
\end{align*}
So the overall cost for all colors is at most $\sum_{C\in \ccal}\left((|C|-1)\cdot \omega + 2|V(G)|\right)$, which is the same as $\sum_{C\in \ccal}\left((|C|-1)\cdot \omega\right) + 2k|V(G)|$. By our choice of $\omega$, we have $2k|V(G)| < \omega$, thus the overall cost does not exceed our chosen budget $b$.

These delaying operations allow $s$ to reach all vertices in \gcal:
\begin{itemize}
    \item For $A\in \ccal$, all $\topv{a}_i$ are reachable because they are on $P_s$.
    \item For $A\in \ccal$, all $\botv{a}_i$ are reachable either via $P_A^\rightarrow$ if $IS(A)\leq i$, or via $P_A^\leftarrow$ if $IS(A)>i$.
    \item For $e=\{a_i, b_j\}\in E$, at least one of $a_i$ or $b_j$ is not in $IS$, wlog let it be $a_i$ (so $IS(A)\neq i$). Then the vertex $e_{a_i, b_j}$ is reachable via $P_A^\rightarrow$ if $IS(A) < i$, or via $P_A^\leftarrow$ if $IS(A)>i$.
\end{itemize}
As every vertex is reachable, and the delaying operations do not exceed the budget, this proves the lemma.
\end{proof}

Towards the other direction, we will determine some properties about a potential solution $\gcal'$ of the \dprm instance that makes all vertices reachable by $s$. The following observation follows immediately from the fact that $s$ only has an outgoing edge of label $0$ and edges can only be delayed.
\begin{observation}\label{obs:delay_w1_by_k_only_nonneg_labels_usable}
    Any edge with a negative label in $\gcal'$ cannot be used by a path from $s$.
\end{observation}

For each $A\in \ccal$, let $\topv{a}_{r(A)}$ be the first such vertex for which its outgoing edge of $P_A^\rightarrow$ is delayed to be non-negative. Likewise, let $\topv{a}_{\ell(A)}$ be the last such vertex for which its outgoing edge of $P_A^\leftarrow$ is delayed to be non-negative.

\begin{lemma}\label{lem:delay_w1_by_k_no_too_big_gap}
    For each $A\in \ccal$ we have $\ell(A) \geq r(A)-1$.
\end{lemma}
\begin{proof}
If there was an $m$ with $\ell(A) < m < r(A)$, then $\botv{a}_m$ would only have incoming edges of negative label and by \Cref{obs:delay_w1_by_k_only_nonneg_labels_usable} it would not be reachable by $s$.
\end{proof}

\begin{lemma}\label{lem:delay_w1_by_k_cost_by_boundaries}
    For each $A\in \ccal$, the cost of delaying operations applied to $P_A^\rightarrow$ and $P_A^\leftarrow$ together is at least $(\ell(A)+|A|-r(A))\cdot \omega$.
\end{lemma}

\begin{proof}
    Because the label of $(\topv{a}_{\ell(A)}, \botv{a}_{\ell(A)})$ becomes non-negative while it initially is $-\ell(A)\cdot \omega$, the cost of delaying operations on $P_A^\leftarrow$ must be at least $\ell(A)\cdot \omega$. Likewise, the label of $(\topv{a}_{r(A)}, \botv{a}_{r(A)})$ becomes non-negative while it initially is $-(|A|-r(A))\cdot \omega$, so the cost for delaying $P_A^\rightarrow$ must be at least $(|A|-r(A))\cdot \omega$. Then both together cost at least $(\ell(A)+|A|-r(A))\cdot \omega$.
\end{proof}

With \Cref{lem:delay_w1_by_k_cost_by_boundaries}, we can make the following stronger statement:

\begin{lemma}\label{lem:delay_w1_by_k_ell_and_r_are_adjacent}
    For each $A\in \ccal$, we have $\ell(A)=r(A)-1$.
\end{lemma}
\begin{proof}
    From \Cref{lem:delay_w1_by_k_no_too_big_gap} we know that $\ell(A)\geq r(A)-1$. Therefore, the cost for the two base-paths associated with $A$ is by \Cref{lem:delay_w1_by_k_cost_by_boundaries} at least $(\ell(A)+|A|-r(A))\cdot \omega \geq (r(A)-1+|A|-r(A))\cdot \omega=(|A|-1)\cdot \omega$. This holds for all colors, so the overall cost is at least $\sum_{C\in \ccal}((|C|-1)\cdot \omega)$. If we now had $\ell(A)>r(A)-1$, then the cost for $A$ would be at least $|A|\cdot \omega$, increasing the previous lower bound on the cost to $\sum_{C\in \ccal}((|C|-1)\cdot \omega)+\omega$. As this would exceed our chosen budget $b$, this is not possible. Therefore $\ell(A)$ must be equal to $r(A)-1$.
\end{proof}

Now we can prove the other direction of the reduction and the overall theorem:

\begin{proof}[Proof of \Cref{thm:delay_w1_by_k}]

    The construction can be computed in polynomial time as the size size of $\gcal$ is linear in $|G|$. Furthermore, because we construct a \kpathgraph{k'} with $k'=2k+1$, the \Wone-hardness transfers.
    
    \bigparagraph{\mcis $\Rightarrow$ \dprm.} This direction is proven by \Cref{lem:delay_w1_by_k_IS_implies_dpmr}.

    \bigparagraph{\dprm $\Rightarrow$ \mcis.}
    We show that $IS=\{a_{r(A)}\in A \vert A \in \ccal\}$ is a multicolored independent set in $G$. For this we show for each $e=\{a_i,b_j\}\in E$ that at least one of $a_i$ and $b_j$ is not in $IS$. Because $e_{a_i,b_j}$ is reached by $s$ in $\gcal'$, one of its incoming edges must be delayed to be non-negative. Wlog let it be the edge of $P_A^\rightarrow$ or $P_A^\leftarrow$. If it is $P_A^\rightarrow$, then $e_{a_i,b_j}$ is after $a_{r(A)}$ on $P_A^\rightarrow$, so $i>r(A)$. If it is $P_A^\leftarrow$, then $e_{a_i,b_j}$ is after $a_{\ell(A)}$ on $P_A^\leftarrow$, so $i\leq \ell(A) < r(A)$ by \Cref{lem:delay_w1_by_k_ell_and_r_are_adjacent}. In both cases $i\neq r(A)$, so $a_i$ is not in $IS$.
\end{proof}

\fi
\ifshort
\bigparagraph{\mcis $\Rightarrow$ \dprm.} Given an independent set $IS$ in $G$ of size $k$, let $IS(A)$ be the index of the vertex chosen for color $A$. We can delay $P_A^\rightarrow$ after $\topv{a_{IS(A)}}$, and delay $P_A^\leftarrow$ after $\topv{a_{IS(A)-1}}$, by so much that switching to them from $P_s$ in that vertex is possible. This costs $(|A|-1)\cdot \omega$ plus some additive non-$\omega$-term, which is within the budget. Also, all vertices become reachable, which for $e$-vertices relies on $IS$ being an independent set.

\bigparagraph{\dprm $\Rightarrow$ \mcis.} Given a solution to \dprm where $s$ can reach all vertices, we can show that for each color $A$, $P_A^\rightarrow$ and $P_A^\leftarrow$ must be delayed to have at least some positive labels. Let $(\topv{a_\ell}, \botv{a_\ell})$ (resp $(\topv{a_r}, \botv{a_r})$) be the first edge along $P_A^\leftarrow$ (resp $P_A^\rightarrow$) to receive a positive label. We proof that $\ell = r-1$ due to the budget constraint. We can then show that choosing $a_r$ for $A$ and doing likewise for every other color yields an independent set, which relies on $s$ being able to reach all $e$-vertices.

\medskip
\fi

While delaying a path makes it easier to switch onto that path, advancing a path makes it easier to switch from that path onto other paths. This changes the strategy for making a specific \switchvertexset temporal: Where delaying a path can only make the switch onto that path temporal, advancing a path can make many switches from that path to other paths temporal.

In turn, the reduction used in \Cref{thm:delay_w1_by_k} does not transfer to the case of advancing. We leave \aprm parameterized by $k$ as an open question, though we do give the following starting point: The strategy of our other positive results (guessing an \switchpathtree \sptree and finding shifting operations to maximize the reachability respecting \sptree) would likely not work for this problem, as solving \aprm \emph{for a specific} \switchpathtree is \Wone-hard.

\begin{theorem}\label{thm:fixedSPT_w1_by_k}
    Given an \switchpathtree \sptree, \dasprm for maximizing the $\sptree$-reachability of the source is \Wone-hard when parameterized by $k$.
\end{theorem}

This reduction is again from \MCIS and functions very similarly to the one for \Cref{thm:delay_w1_by_k}. Now we adapt the construction in a way where any switch can be temporalized both via advancing an edge or delaying an edge, for the same cost. \iflong The downside of this construction is that switches outside of the intended ones may be possible and even cheaper. Because we only consider the \sptree-reachability for a specific \switchpathtree \sptree, we can choose \sptree such that theses unintended switches do not benefit the \sptree-reachability.\fi
\ifshort
See \Cref{fig:fixedSPT_w1_by_k_main} for the overall construction and the chosen $\sptree$. The budget is $b=\sum_{C\in \ccal}(|C|-1)\cdot \omega$.
\fi

\iflong
\bigparagraph{Construction.}
Given a $k\in \Na$ and an undirected static graph $G=(\bigcup_{C\in\ccal}C, E)$ for the \mcis instance, we again describe a color-gadget where specific advancing operations correspond to choosing a node. Afterwards we show how to represent the edges $E$ of $G$ to restrict the available choices.

We again will use large time labels that are multiples of some large $\omega$. Choosing $\omega=|V(G)|^4$ is again sufficient.

\newcommand{\midv}[1]{\ensuremath{#1^\circ}}

We begin with an empty base-path $P_s$ starting in a vertex $s$. Construct a gadget for each color $A=\{a_1, \dots a_n\}\in \ccal$ as follows:\begin{enumerate}
    \item For each $0\leq i \leq n$, create three vertices $\topv{a_i}, \botv{a_i}, \midv{a_i}$. Also create vertices $\topv{a_s}$ and $\botv{a_s}$.
    \item Create a base-path $P_A^\top$ as $(\topv{a_s}, \topv{a_0}, \dots \topv{a_n})$, and a base path $P_A^\bot$ as $(\botv{a_s}, \botv{a_n}, \dots \botv{a_0})$. Label both path with increasing multiples of $\omega$, starting with label $0\omega$.
    \item Create a base-path $P_A^\leftarrow$ as $(\topv{a_{n-1}}, \midv{a_{n-1}}, \topv{a_{n-2}}, \midv{a_{n-2}}, \dots \topv{a_0}, \midv{a_0})$. 
    \item Likewise create a base-path $P_A^\rightarrow$ as $(\botv{a_1}, \midv{a_1}, \botv{a_2}, \midv{a_2}, \dots \botv{a_n}, \midv{a_n})$.
    \item Label both $P_A^\leftarrow$ and $P_A^\rightarrow$ incrementally starting with $1$.
    \item Extend $P_s$ to visit $\topv{a_s}$ and $\botv{a_s}$. Edges of $P_s$ are labeled incrementally starting with $-\omega$.
\end{enumerate}

\begin{figure}[h]
    \centering
    \includegraphics[width=1\linewidth]{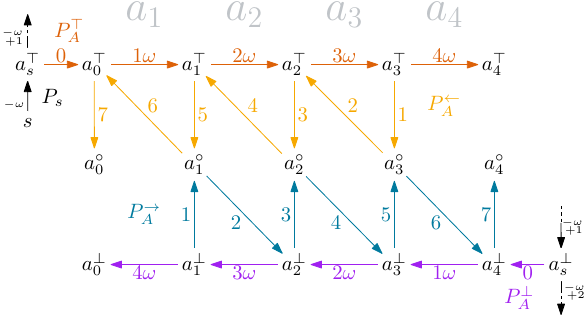}
    \caption{The color gadget for a color $A\in \ccal$ where $A=\{a_1, a_2, a_3, a_4\}$.
    }
    \label{fig:fixedSPT_w1_by_k_color_gadget}
\end{figure}

See \Cref{fig:fixedSPT_w1_by_k_color_gadget} for an example of a color-gadget. Again, a node $a_i\in A$ corresponds to the gap between the vertex-triplet $\{\topv{a_{i-1}}, \midv{a_{i-1}}, \botv{a_{i-1}}\}$ and $\{\topv{a_i}, \midv{a_i}, \botv{a_i}\}$. Choosing $a_i$ will correspond to creating a temporal switch from $P_A^\top$ to $P_A^\leftarrow$ at the vertex $\topv{a_{i-1}}$, and creating a temporal switch from $P_A^\bot$ to $P_A^\rightarrow$ at the vertex $\botv{a_i}$. Such a switch can be created using delaying, advancing, or both, and costs the same in all cases.

To represent the edges of $G$, we will create additional vertices and insert them into existing base-paths. For each $e=\{a_i, b_j\}\in E$:
\begin{enumerate}
    \item Create a new vertex $e_{a_i,b_j}$.
    \item If $i>1$, insert $e_{a_i,b_j}$ into $\subpath{ P_A^\rightarrow}{\midv{a}_{i-1}}{\bot{a}_i}$: Replace an arbitrary edge $(u,v,t)$ on $\subpath{ P_A^\rightarrow}{\midv{a}_{i-1}}{\botv{a}_i}$ with $(u, e_{a_i,b_j}, t)$ and $(e_{a_1,b_j}, v, t+1)$, and increase subsequent labels of $P_A^\rightarrow$ by $1$.
    \item Likewise, if $i<|A|$, insert $e_{a_i,b_j}$ into $\subpath{P_A^\leftarrow}{\midv{a}_i}{\topv{a}_{i-1}}$.
    \item In the same way, insert it into $P_B^\rightarrow$ and $P_B^\leftarrow$.
\end{enumerate}

See \Cref{fig:fixedSPT_w1_by_k_main} for an example.
\fi
\begin{figure}[h]
    \centering
    \includegraphics[width=0.9\linewidth]{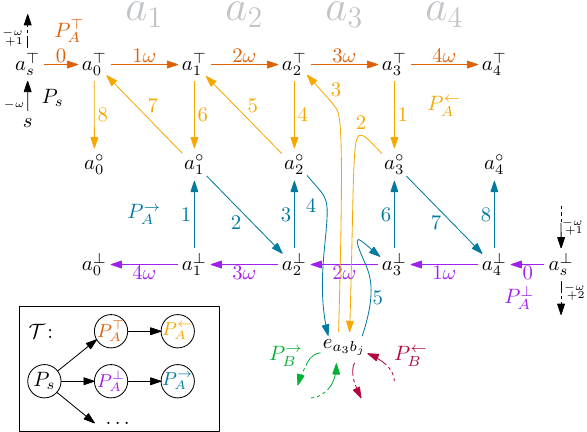}
    \caption{The color-gadget for color $A$ and the representation of an edge between $a_3\in A$ and some $b_j\in B$. Below is the \switchpathtree \sptree for which we are to optimize the reachability.}
    \label{fig:fixedSPT_w1_by_k_main}
\end{figure}
\iflong
For the budget we choose $b=\sum_{C\in \ccal}(|C|-1)\cdot \omega$. Last we need to specify an \switchpathtree \sptree. We choose \sptree via the edge-set $\{(P_s, P_A^\top), (P_s, P_A^\bot), (P_A^\top, P_A^\rightarrow), (P_A^\bot, P_A^\leftarrow) \mid A\in \ccal\}$.

Now we can show that a multicolored independent set of size $k$ exists in $G$ if and only if it is possible to make $s$ reach \emph{all} vertices in \gcal via \sptree-respecting paths using only delaying, as well as only advancing, as well as any shifting operations within the given budget. 

\begin{lemma}\label{lem:fixedSPT_w1_by_k_IS_implies_dpmr}
    If $G$ admits a multicolored independent set, then an optimal solution to \dasprm on $\gcal$ allows $s$ to reach all vertices via \sptree-respecting paths.
\end{lemma}
\begin{proof}
Let $IS$ be a multicolored independent set of size $k$ in $G$, and for each color $A\in \ccal$, let $IS(A)$ be the index of the node of $A$ that is chosen in $IS$. Based on this, we apply the following operations in \gcal:
\begin{itemize}
    \item For \dprm and \sprm: Delay $P_A^\rightarrow$ at $\botv{a_{IS(A)}}$ by $(|A|-IS(A))\cdot \omega$. For \aprm, advance $P_A^\bot$ at $\botv{a_{IS(A)}}$ by $(|A|-IS(A))\cdot \omega$. In both cases, the switch $(\botv{a_{IS(A)}}, P_A^\bot, P_A^\rightarrow)$ becomes temporal.

    \item For \dprm and \sprm: Delay $P_A^\leftarrow$ at $\topv{a_{IS(A)-1}}$ by $(IS(A)-1)\cdot \omega$. For \aprm, advance $P_A^\top$ at $\topv{a_{IS(A)-1}}$ by $(IS(A)-1)\cdot \omega$. In both cases, the switch $(\topv{a_{IS(A)-1}}, P_A^\top, P_A^\leftarrow)$  becomes temporal.

\end{itemize}

The cost of both operations for the color $A$ adds up to $(|A|-IS(A))\cdot \omega+(IS(A)-1)\cdot \omega$, which is exactly $(|A|-1)\cdot \omega$
So the overall cost for all colors is $\sum_{C\in \ccal}\left((|C|-1)\cdot \omega \right)$, which is exactly our budget $b$.

These delaying operations allow $s$ to reach all vertices in \gcal:
\begin{itemize}
    \item For $A\in \ccal$, all $\topv{a_i}$ (resp. $\botv{a_i}$) are reachable as they are on $P_A^\top$ (resp. $P_A^\bot$) and the switch $(\topv{a_s}, P_s, P_A^\top)$ (resp. $(\botv{a_s}, P_s, P_A^\bot)$) remains temporal.
    \item For $A\in \ccal$, all $\midv{a_i}$ are reachable either via $P_A^\rightarrow$ if $IS(A)\leq i$, or via $P_A^\leftarrow$ if $IS(A)>i$.
    \item For $e=\{a_i, b_j\}\in E$, at least one of $a_i$ or $b_j$ is not in $IS$, wlog let it be $a_i$ (so $IS(A)\neq i$). Then the vertex $e_{a_i, b_j}$ is reachable via $P_A^\rightarrow$ if $IS(A) < i$, or via $P_A^\leftarrow$ if $IS(A)>i$.
\end{itemize}
As every vertex is made reachable without exceeding the budget, this proves the lemma.
\end{proof}

Towards the other direction, we will determine some properties about a potential solution $\gcal'$ of the \dasprm instance that makes all vertices reachable by $s$.

For each $A\in \ccal$, let $\botv{a_{r(A)}}$ be the first vertex along $P_A^\rightarrow$ for which $(\botv{a_{r(A)}}, P_A^\bot, P_A^\rightarrow)$ is a temporal switch in $\gcal'$.
Because $\midv{a_0}$ is reached and only has an incoming edge from $P_A^\leftarrow$, and our chosen \sptree only permits switches onto $P_A^\leftarrow$ from $P_A^\top$, at least one such switch must exist.

Likewise, let $\topv{a_{\ell(A)}}$ be first vertex along $P_A^\leftarrow$ for which $(\topv{a_{\ell(A)}}, P_A^\top, P_A^\leftarrow)$ is a temporal switch in $\gcal'$. The existence is guaranteed with a similar argument as before, now via $\midv{a_{|A|}}$.

\begin{lemma}\label{lem:fixedSPT_w1_by_k_no_too_big_gap}
    For each $A\in \ccal$ we have $\ell(A) \geq r(A)-1$.
\end{lemma}
\begin{proof}
If there was an $m$ with $\ell(A) < m < r(A)$, then $\midv{a}_m$ would be on $P_A^\leftarrow$ before the first temporal switch onto it, and it would be on $P_A^\rightarrow$ before the first temporal switch onto it. Then there cannot be a \sptree-respecting path from $s$ to $\midv{a_m}$, so it would not be in the \sptree-reachability of $s$.
\end{proof}

\begin{lemma}\label{lem:fixedSPT_w1_by_k_cost_by_boundaries}
    For each $A\in \ccal$, the combined cost of operations applied to the paths related to $A$ is at least $(\ell(A)+|A|-r(A))\cdot \omega - 2|V(G)|^2$.
\end{lemma}
\begin{proof}
    The edge of $P_A^\bot$ going into $\botv{a_{r(A)}}$ has label $(|A|-r(A))\cdot \omega$, while the edge of $P_A^\rightarrow$ going out of $\botv{a_{r(A)}}$ has label at most $|V(G)|^2$. Because $(\botv{a_{r(A)}}, P_A^\bot, P_A^\rightarrow)$ is a temporal switch in $\gcal'$, the total advance applied to $P_A^\bot$ plus the total delay applied to $P_A^\rightarrow$ must be at least $(|A|-r(A))\cdot \omega-|V(G)|^2$.

    The edge of $P_A^\top$ going into $\topv{a_{\ell(A)}}$ has label $\ell(A)\cdot \omega$, while the edge of $P_A^\leftarrow$ going out of $\topv{a_{\ell(A)}}$ has label at most $|V(G)|^2$. Because $(\topv{a_{\ell(A)}})$ is a temporal switch in $\gcal'$, the total advance applied to $P_A^\top$ plus the total delay applied to $P_A^\leftarrow$ must be at least $\ell(A)\cdot \omega - |V(G)|^2$.
    
    Combined, the shifting operations applied to all four paths related to $A$ must be at least $(\ell(A)+|A|-r(A))\cdot \omega - 2|V(G)|^2$.
\end{proof}

We can use \Cref{lem:fixedSPT_w1_by_k_no_too_big_gap} to also bound the difference between $\ell(A)$ and $r(A)$ in the other direction:

\begin{lemma}\label{lem:fixedSPT_w1_by_k_ell_and_r_are_adjacent}
    For each $A\in \ccal$, we have $\ell(A)= r(A)-1$.
\end{lemma}
\begin{proof}
    
    By \Cref{lem:fixedSPT_w1_by_k_no_too_big_gap} we know that $\ell(A)\geq r(A)-1$. Therefore, the cost for the two base-paths associated with $A$ is by \Cref{lem:fixedSPT_w1_by_k_cost_by_boundaries} at least $(\ell(A)+|A|-r(A))\cdot \omega - 2|V(G)|^2 \geq (r(A)-1+|A|-r(A))\cdot \omega - 2|V(G)|^2 = (|A|-1)\cdot \omega - 2|V(G)|^2$. This holds for all colors, so the overall cost is at least $\sum_{C\in \ccal}((|C|-1)\cdot \omega - 2|V(G)|^2)$. If we now had $\ell(A)>r(A)-1$, then the cost for $A$ would be at least $|A|\cdot \omega - 2|V(G)|^2$, increasing the previous lower bound on the cost to $\sum_{C\in \ccal}((|C|-1)\cdot \omega - 2|V(G)|^2)+\omega = \sum_{C\in \ccal}((|C|-1)\cdot \omega) - |\ccal|\cdot 2|V(G)|^2+\omega$. By our choice of $\omega$ we have $\omega > |\ccal| \cdot 2|V(G)|^2$, so this would exceed our chosen budget $b$. As this is a contradiction, $\ell(A)$ must instead be equal to $r(A)-1$.
\end{proof}

Now we can prove the other direction of the reduction and the overall theorem:
\begin{proof}[Proof of \Cref{thm:fixedSPT_w1_by_k}]

    The construction can be computed in polynomial time as the size size of $\gcal$ is linear in $|G|$. Furthermore, because we construct a \kpathgraph{k'} with $k'=4k+1$, the \Wone-hardness transfers.
    
    \bigparagraph{\mcis $\Rightarrow$ \dprm.} This direction is proven by \Cref{lem:fixedSPT_w1_by_k_IS_implies_dpmr}.

    \bigparagraph{\dprm $\Rightarrow$ \mcis.}
    We show that $IS=\{a_{r(A)}\in A \vert A \in \ccal\}$ is a multicolored independent set in $G$. For this we show for each $e=\{a_i,b_j\}\in E$ that at least one of $a_i$ and $b_j$ is not in $IS$. Because $e_{a_i,b_j}$ is reached by $s$ in $\gcal'$, there must be a \sptree-respecting path to it. Wlog let its last edge be the incoming edge from either $P_A^\rightarrow$ or $P_A^\leftarrow$. If it is $P_A^\rightarrow$, then $e_{a_i,b_j}$ is after $a_{r(A)}$ on $P_A^\rightarrow$, so $i>r(A)$. If it is $P_A^\leftarrow$, then $e_{a_i,b_j}$ is after $a_{\ell(A)}$ on $P_A^\leftarrow$, so $i\leq \ell(A) < r(A)$ by \Cref{lem:fixedSPT_w1_by_k_ell_and_r_are_adjacent}. In both cases $i\neq r(A)$, so $a_i$ is not in $IS$.
\end{proof}
\fi

The reason the above reduction does not work in the general case is because other switches than those allowed by $\sptree$ may be more useful without corresponding to an independent set in $G$. For example, in \Cref{fig:fixedSPT_w1_by_k_main}, advancing the entire $P_A^\rightarrow$ by $6$ allows switching from $P_A^\rightarrow$ to $P_A^\leftarrow$ at a much lower cost that the intended switch from $P_A^\top$.

\subsection{Positive Results}\label{sec:positive_results}

In this section, we present our positive results. Specifically, we give two \XP algorithms for the three modification variants, parameterized by $b$ and by $k$, respectively, and two \FPT algorithms parameterized by $b+k$: a faster and simpler algorithm for delay-only modifications, and a second, more involved algorithm for all three modification variants.

When parameterized by the budget $b$, an \XP algorithm follows readily: for each unit of the budget, one guesses which edge to modify with it, yielding a bounded enumeration.

\begin{theorem}\label{thm:budgeted_XP_by_b}
    \dasprm is in \XP when parameterized by $b$, and can be solved in $\bigoh(|\gcal|^b)$ time.
\end{theorem}
\iflong
\begin{proof}
    Given a \kpathgraph{k} and a source $s$, guess for each unit of budget which edge to apply it to, or whether to not use this unit at all. This gives $|\ecal|+1$ choices per unit of budget. For shifting, there are $2|\ecal|+1$ choices as you can either delay any edge, or advance any edge, or not use that unit of budget. This enumerates all possible resulting \kpathgraphs{k} in $\bigoh((2|\ecal|+1)^b)=\bigoh(|\gcal|^b)$ time. Among these, return the $\gcal'$ with maximum $|R^{\gcal'}(s)|$.
\end{proof}
\fi
When parameterized by the number of base-paths $k$, the algorithm is more involved. We enumerate candidate \switchvertexsets{} and, for each, check whether its switches can all be made temporal within budget using an integer linear program (ILP). By \Cref{lem:canonical_shift_form} we only need to consider operations applied to the edges referenced by a switch. The key ingredient is a compact description of how delays and advances propagate along each base-path. This relies heavily on existing waiting times:

\begin{definition}[slack]
    We denote the waiting time between two consecutive edges $(e_i,t_i)$ and $(e_{i+1},t_{i+1})$ on a base-path by
    \(\mathsf{slack}((e_i,t_i),(e_{i+1},t_{i+1})):= t_{i+1} - t_i - 1.\)
    
    More generally, for two vertices $u \leq_P v$ on a base-path $P$, the cumulative slack between them is the total waiting time along $P$ from $u$ to $v$:
    \[\mathsf{slack}(u, v) := \sum_{\text{consecutive } (e_i,t_i),\,(e_{i+1},t_{i+1}) \text{ on } P[u:v]} \mathsf{slack}\big((e_i,t_i),(e_{i+1},t_{i+1})\big).\]
\end{definition}

While a shifting operation propagates forwards or backwards along a base-path, the propagation amount is effectively absorbed by slack, or passes though completely if there is no slack. The propagating amount applied to some edge through an operation on another edge of the same base-path thus only depends on the slack between them. Precomputing this slack allows us to abstract away all non-switch vertices and describe all propagating effects in an ILP of a size only dependent on $k$.

\begin{theorem}\label{thm:budgeted_XP_by_k}
    \dasprm is in \XP when parameterized by $k$, and can be solved in $\bigoh(|V|^k\cdot (\log k)^{\bigoh(k)}\cdot \poly(|\gcal|))$ time.
\end{theorem}
\iflong
\begin{proof}
    For a given \kpathgraph{k} $\gcal = \biguplus_{i\in[k]} P_i$ and source $s$, we enumerate all valid \switchvertexsets \svset. Since each of the at most $k$ switch-vertices is chosen from $V$, there are $\bigoh(\lvert V\rvert^k)$ candidates, which can be enumerated in \XP time parameterized by $k$.
        By \Cref{lem:svset_respecting_reachability}, once all switches of a fixed $\svset$ are temporal, the maximum reachability $|R^{\gcal}_{\svset}(s)|$ is determined and computable in $\poly(|\gcal|)$. It therefore suffices to determine the \emph{minimum cost} required to make all switches of $\svset$ temporal.
        For that we are going to use an ILP. Let $\beta(\svset)$ be the optimum value of this ILP. If $\beta(\svset)\leq b$, then we record $|R^{\gcal}_{\svset}(s)|$; otherwise, we discard $\svset$. Ultimately, we return the maximum recorded reachability over all \svset with $\beta(\svset)\leq b$.

    For every switch $\sw = (v, \swPfrom, \swPto) \in \svset$, we introduce two non-negative variables~$d_{\sw}$ (delay applied to \swPto at \sw) and~$a_{\sw}$ (advance applied to \swPfrom at \sw). 
    Here we intentionally deviate from the signed shifting convention from the formal definition: in this ILP, $a_{\sw}$ is nonnegative, denotes the \emph{magnitude} of the advance and is subtracted from the considered time label. 

    \begin{figure}[t]
        \centering
        \includegraphics[width=0.85\linewidth]{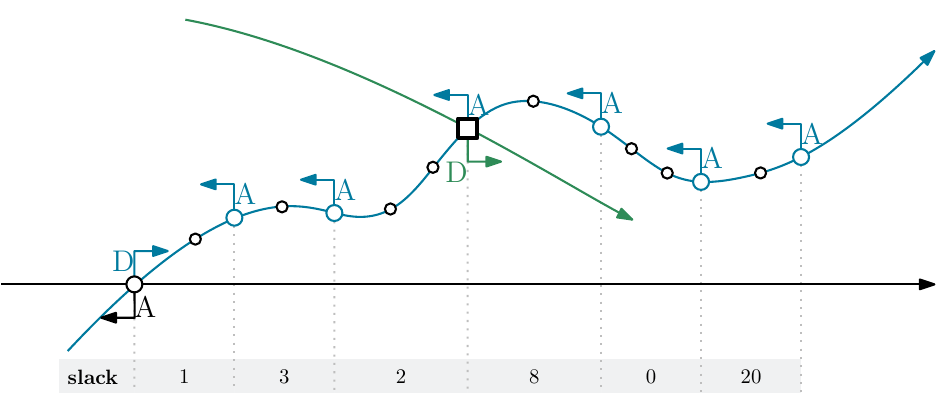}
        \caption{Illustration of the propagation along a base-path~$P$ (blue). The entry delay~$d_{\sw_0^P}$ propagates \emph{forward} (to the right) from the unique switch onto~$P$, being absorbed by slack and by advances encountered along the way. Advances at off-switches propagate \emph{backward} (to the left), likewise absorbed by slack. The net propagation that reaches a switch-vertex~$v$ from both directions determines the temporality constraint for every switch at~$v$. The bottom row indicates the cumulative slack $\mathsf{slack}(v,v^+)$ between consecutive switch-vertices on~$P$.}
        \label{fig:placeholder}
    \end{figure}
    \bigparagraph{1. Structure of switches per base-path.}\quad
    Recall from \Cref{lem:only_one_segment_of_each_basepath_needed} that for each base-path $P \in \pcal \setminus \{P_s\}$, there is exactly one switch \emph{onto} $P$, namely $\sw_0^P = (v_0^P, \cdot, P)$, which may require a delay $d_{\sw_0^P} \geq 0$. All remaining switches involving $P$ are \emph{off} $P$, of the form $\sw = (v, P, \cdot)$, each potentially requiring an advance $a_{\sw} \geq 0$. These delays and advances do not act in isolation: they propagate along $P$ and interact with the slack between edges. \Cref{fig:placeholder} illustrates this structure; the following items make it precise.
    \begin{itemize}
        \item Any delay $d_{\sw_0^P}$ at $\sw_0^P$ propagates \emph{forwards} along $P$ but is absorbed by slack and by any advances along the way.
            To track this in the ILP, we use variables $D^P_v \geq 0$ for each vertex $v$ on $P$ that is involved in a switch, with
            \begin{align}
                D^P_{v_0^P} = d_{\sw_0^P}, \qquad
                D^P_v = D^P_{v^-} - \mathsf{slack}(v^-, v) - \sum_{\sw = (v,\, P,\, \cdot)} a_{\sw}, \label{eq:delay_propagation}
            \end{align}
            where $v^-$ is the immediate predecessor of $v$ among switch-vertices on $P$.
            The delay that reaches vertex $v$ via propagation\footnote{The propagation variables $D^P_v$ and $A^P_v$ are indexed by \emph{vertex and path}, since propagation is a property of a position along $P$. By contrast, the decision variables $a_{\sw}$ remain \emph{switch}-indexed: a vertex $v$ may carry multiple switches off the same path $P$ (e.g.\ $\sw=(v,P,Q)$ and $\sw'=(v,P,R)$), each with its own advance. Writing $a^P_v$ would be ambiguous. The two levels are connected by the sum $\sum_{\sw=(v,P,\cdot)} a_{\sw}$ in the constraints.} along $P$ is $D^P_v$. \\
            Note that by \Cref{cor:delay_and_advance_should_never_collide}, in a budget-minimal solution one would not advance along $P$ at $v$ if a delay propagation along $P$ reaches $v$. So for each path $P$, at every switch-vertex $v$ on $P$, either $D^P_{v} = 0$ or $a_{\sw} = 0$ for all $\sw = (v, P, \cdot)$. It is still possible to simultaneously delay $\swPto$ and advance $\swPfrom$ at the same switch.
        \item Any advances at switches $\sw = (v, P, \cdot)$ propagate \emph{backwards} along $P$ but are absorbed by any slack along the way.
            To track this in the ILP, we use variables $A^P_v \geq 0$ for each vertex $v$ on $P$ that is involved in a switch, with
            \begin{align}
                A^P_{v_n} = \sum_{\sw=(v_n,\,P,\,\cdot)} a_{\sw}, \qquad A^P_v = \sum_{\sw=(v^+,\,P,\,\cdot)} a_{\sw} + A^P_{v^+} - \mathsf{slack}(v, v^+), \label{eq:advance_propagation}
            \end{align}
            where $v_n$ is the last switch-vertex on $P$ and $v^+$ is the immediate successor of $v$ among switch-vertices on $P$. The advance reaching vertex $v$ along $P$ is $A^P_v$.
        \end{itemize}

    \bigparagraph{2. The temporality constraint for each switch.}\quad
    For each $\sw = (v, \swPfrom, \swPto) \in \svset$, the \emph{temporality condition} $\lambda'(\swEfrom{\sw}) + 1 \leq \lambda'(\swEto{\sw})$ must hold after shifting. {Using the propagation variables $D^P_v$ and $A^P_v$ introduced above, the shifted edge-labels are:}
    \begin{align}
        \lambda'(\swEfrom{\sw}) &= \lambda(\swEfrom{\sw}) - a_{\sw} {+ D^{\swPfrom}_v - A^{\swPfrom}_v}, \label{eq:fromProp}\\
    \intertext{and, analogously,}
        \lambda'(\swEto{\sw}) &= \lambda(\swEto{\sw}) + d_{\sw} {- A^{\swPto}_v}. \label{eq:toProp}
    \end{align}
    {Substituting \Cref{eq:fromProp} and \Cref{eq:toProp} into the temporality condition directly gives:}
    \begin{equation}\label{eq:lp_temporality}
        {
        \lambda(\swEfrom{\sw}) - a_{\sw} + D^{\swPfrom}_v - A^{\swPfrom}_v + 1
        \;\leq\;
        \lambda(\swEto{\sw}) + d_{\sw} - A^{\swPto}_v.}
    \end{equation}

    \bigparagraph{3. The ILP.}\quad
    Combining all this, the cost-minimization ILP for a fixed $\svset$ is then:
    \begin{align*}
        \begin{array}{llll}
        \text{minimize}
            & \displaystyle\sum_{\sw \in \svset} (d_{\sw} + a_{\sw})
                & & \\[10pt]
        \text{over}
            & d_{\sw},\; a_{\sw}
                & \forall\,\sw \in \svset & \\[5pt]
            & {D^P_v,\; A^P_v}
                & {\forall\, P \in \pcal,\; v \in V(P) \cap V(\svset)} & \\[10pt]
        \text{subject to}
            & \eqref{eq:lp_temporality}
                & \forall\, \sw \in \svset & \text{(temporality)} \\[5pt]
            & {\eqref{eq:delay_propagation}~and~\eqref{eq:advance_propagation}}
                & {\forall\, P \in \pcal,\; v \in V(P) \cap V(\svset)} & {\text{(propagation)}} \\[5pt]
            & {d_{\sw},\; a_{\sw},\; D^P_v,\; A^P_v \geq 0}
                & {Q\forall\, \sw \in \svset,\; P \in \pcal,\; v} &
        \end{array}
    \end{align*}
    This ILP has $\bigoh(k)$ variables and $\bigoh(k)$ constraints, since
    $|\svset| \leq k$ and the total number of switch-vertices across all base-paths is $\bigoh(k)$.
    It can therefore be solved in $(\log k)^{\bigoh(k)}\cdot \poly(|\gcal|)$ for example by using the algorithm by Reis and Rotvoss~\cite{reis_subspaceflatness_2023} time for any
    fixed $\svset$.

    \bigparagraph{4. Correctness.}\quad
    First assume that the ILP for some fixed $\svset$ has optimum value $\beta(\svset)\leq b$. Then the values $d_{\sw}$ and $a_{\sw}$ describe delaying and advancing operations of total cost $\beta(\svset)$. Moreover, \eqref{eq:lp_temporality} guarantees that every switch of $\svset$ is temporal after applying these shifts in any order, and thus by \Cref{lem:svset_respecting_reachability} we can realize reachability $|R^{\gcal}_{\svset}(s)|$ within budget $b$.

    Conversely, assume there exists a sequence of shifting operations of cost at most $b$ that makes all switches of $\svset$ temporal. Among all such sequences, choose one of minimum cost. Then by \Cref{cor:delay_and_advance_should_never_collide}, no previously advanced edge is delayed again later, and no previously delayed edge is advanced again later. Hence, the order of the sequences does not matter and the propagation of delays and advances along each base-path is determined solely by the local delays $d_{\sw}$, local advances $a_{\sw}$, and the available slack, and is therefore captured by the propagation constraints \eqref{eq:delay_propagation} and \eqref{eq:advance_propagation}. Since all switches are temporal in the resulting graph, \eqref{eq:lp_temporality} is satisfied as well. Therefore the ILP has a solution of value at most $b$, implying $\beta(\svset)\leq b$.

    As a result, $\beta(\svset)\leq b$ holds if and only if $\svset$ can be realized within budget $b$. Therefore, the maximum recorded reachability over all \svset with $\beta(\svset)\leq b$ equals the optimum.

    Note that in the case where we are only allowed to delay, we can simply remove the advance variable and consideration from the ILP. Analogously, for only advance.
\end{proof}
\fi

\newcommand{\lDel}[2]{\ensuremath{\delta_{#2}}\xspace}
\newcommand{\lAdv}[2]{\ensuremath{\alpha_{#2}}\xspace}
\newcommand{\labDiff}[2]{\ensuremath{\ell_{#2}}\xspace}
\newcommand{\propFrom}[2]{\ensuremath{\Delta_{#2}(#1)}}
\newcommand{\propTo}[2]{\ensuremath{\Delta_{#2}(#2)}}
\newcommand{\switchtuple}[2]{(\lDel{#1}{#2},\ \lAdv{#1}{#2},\ \propFrom{#1}{#2},\ \propTo{#1}{#2},\ \labDiff{#1}{#2})\xspace}

\newcommand{\DelOut}[1]{\ensuremath{D_{#1}^{\rightarrow}}\xspace}         %
\newcommand{\DelIn}[1]{\ensuremath{{}^{\rightarrow}\!D_{#1}}\xspace}       %
\newcommand{\AdvOut}[1]{\ensuremath{{}^{\leftarrow}A_{#1}}\xspace}           %
\newcommand{\AdvIn}[1]{\ensuremath{A_{#1}{}^{\leftarrow}}\xspace}        %
\newcommand{\delQ}[1]{\ensuremath{\delta_{#1}}\xspace}                     %
\newcommand{\advInQ}[1]{\ensuremath{\phi_{#1}{}^{\leftarrow}}\xspace}    %
\newcommand{\actAdv}[1]{\ensuremath{\alpha_{#1}}\xspace}                   %

\medskip
We have seen earlier that under any operation, our problem is \Wtwo-hard when parameterized by the budget $b$ (\Cref{thm:budgeted_w2_by_b}), and that for delaying, it is \Wone-hard when parameterized by the number of base-paths $k$ (\Cref{thm:delay_w1_by_k}).
In the following, we show that combining both parameters yields \FPT algorithms for all three modification variants. Note that an FPT run-time by $b+k$ only allows for guessing the \switchpathtree, but not the \switchvertexset as in the \XP algorithm for parameter $k$ above.

{
    We present two algorithms: a simpler and faster one for delaying only (Theorem~\ref{thm:delay_fpt_by_b_and_k}), and a more complicated one that works for all three modification types (Theorem~\ref{thm:all_fpt_by_b_and_k}).
    Both follow the same two-phase structure:
    \begin{description}
        \item[Enumeration phase.]\quad Guess the \switchpathtree{} $\sptree$ and the budget-related values around the switch onto each base-path, that includes particularly the delay or advance that is to be applied.
        \item[Greedy phase.]\quad Given the guessed values, compute the optimal \switchvertexset{} by processing $\sptree$ root-to-leaves and assigning each path the earliest \emph{valid} switch vertex. 
    \end{description}

    For \emph{only delaying}, it suffices to guess $\sptree$ and a single delay amount $\delta_P \in [b]$ per non-source base-path $P$ (yielding $\bigoh(k^k)$ and $\bigoh(b^{k-1})$ choices, respectively).
    The greedy phase then exploits that delaying an edge propagates only \emph{forward} on that base-path. Consequently, the \emph{earliest} switch-vertex onto the base-path $P$ simultaneously maximises the reachable suffix on $P$ and minimises the propagated delay impacting later switches on $\sptree$.
    Both objectives point in the same direction (away from the source-path), so choosing the earliest switch that fits the guessed delay is optimal.
        The resulting running time is $\bigoh\bigl(b^k\cdot k^k \cdot |\gcal|^{\bigoh(1)}\bigr)$.

\begin{theorem}\label{thm:delay_fpt_by_b_and_k}
    \dprm is in \FPT when parameterized by $b+k$, and can be solved in $\bigoh\!\left(b^k\cdot k^k\cdot |\gcal|^{\bigoh(1)}\right)$ time.
\end{theorem}
\iflong
\begin{proof}
    Let $\gcal= \biguplus_{i\in [k]}P_i$ be the input \kpathgraph{k} and let $s$ be a source. Furthermore, let $\gcal^*$ be any optimal solution, that is, a \kpathgraph{k} maximizing $|\rcal^{\gcal^*}(s)|$ among all \kpathgraph{k}s that can be  obtained from $\gcal$ by delaying within the budget $b$. 

    By Lemma~\ref{lem:switchverts_is_enough}, there is an \switchvertexset{} $\svset^*$ such that $\rcal^{\gcal^*}_{\svset^*}(s) = \rcal^{\gcal^*}(s)$ and by definition it implies an \switchpathtree{} $\sptree_{\svset^*}$. 
    By \Cref{lem:spt_enumeration} we can enumerate all possibilities for $\sptree_{\svset^*}$ in $\bigoh(k^{k})$ time.
    Fix one such $\sptree$ and assume that $\sptree=\sptree_{\svset^*}$.
    The argument is structured into 4 steps.

    \bigparagraph{1. Canonical form of an optimal solution.}
    For every non-source base-path $P$, let $(v_P^*, \swPfrom, P)$ be the unique switch onto $P$ in $\svset^*$. By \Cref{lem:canonical_shift_form}, specialized to the delay-only setting, we may assume that the optimal solution applies only one local delay to the temporal edge of $P$ leaving $v_P^*$. In particular, no other temporal edge on $P$ needs to be delayed for the switch onto $P$, and hence each non-source path contributes only this single local delay. 
    As a result, there are at most $k-1$ delayed edges, one for each non-source base-path.

    \bigparagraph{2. Enumeration of the relevant values.}
    For the sake of exposition, we will consider $(\emptyset,P_s)\in E(\sptree)$ as the empty switch onto the root path $P_s$. Now, to be able to compute the switch vertex for each switch $(\swPfrom,\swPto)\in E(\sptree)$, we guess the following value:
    \begin{itemize}
        \item the local delay $\delta_{\swPto}\in [0,b]$ to be applied to $\swPto$ at the switch,
    \end{itemize}
    such that $\sum_{(\swPfrom,\swPto)\in E(\sptree)} \delta_{\swPto}\le b$.
    The number of possibilities for the guessed delays is at most $(b+1)^{k-1}$.

    \bigparagraph{3. Greedy construction.}
    Assume we guessed $\delta_{\swPto}$ for every edge $(\swPfrom,\swPto)\in E(\sptree)$.
    We will iteratively compute an \switchvertexset $\svset$ satisfying the following properties:
    \begin{description}
        \item[\namedlabel{item:sptprime_is_sptstar}{(\switchvertexset 1)}(\switchvertexset 1)] \quad $\sptree_{\svset}=\sptree=\sptree_{\svset^*}$,
        \item[\namedlabel{item:vprime_before_vstar}{(\switchvertexset 2)}(\switchvertexset 2)] \quad for all $(v_{\swPto}, \swPfrom,\swPto)\in \svset$ and the corresponding $(v^*_{\swPto}, \swPfrom,\swPto)\in \svset^*$, the vertex $v_{\swPto}$ is before (or equal to) $v_{\swPto}^*$ on $\swPto$ ($v_{\swPto} \le_{\swPto} v_{\swPto}^*$), and
        \item[\namedlabel{item:svsprime_is_SVS_delaythis}{(\switchvertexset 3)}(\switchvertexset 3)] \quad $\svset$ is a temporal \switchvertexset in the \kpathgraph{k} $\gcal'$ obtained from $\gcal$ by applying, for every $(v_{\swPto}, \swPfrom, \swPto)\in \svset$, the delay $\delta_{\swPto}$ to $e_{\swPto}$.
    \end{description}
    We compute $\svset$ by going root-to-leaves along $\sptree$.
    
    Consider $(s,\emptyset, P_s)$ as the switch onto the root path $P_s$.
    Then, in both \switchvertexset{s} $\svset$ and $\svset^*$, we switch onto $P_s$ at $s$ with delay $\delta_{P_s}=0$, and we can use this as the base case. 

    Suppose that for $(\swPfrom,\swPto)\in E(\sptree)$ we fixed the switch onto $\swPfrom$ as $v_{\swPfrom}$. 
    Let $\swPfrom'$ be the path obtained from $\swPfrom$  by applying delay $\delta_{\swPfrom}$ to the temporal edge $(e'_{\swPfrom}, t'_{\swPfrom})$.
    We call a common vertex $v$ of $\swPfrom$ and $\swPto$ \emph{valid} for the switch from $\swPfrom$ to $\swPto$ if and only if: %
    \begin{description}
        \item[\namedlabel{item:switchoff_before_switchon}{(P1)}(P1)] $v_{\swPfrom} \le_{\swPfrom} v$, and %
        \item[\namedlabel{item:svs_with_delaythis}{(P2)}(P2)]
            $\lambda(e_{from}) + \max\bigl\{0,\ \delta_{\swPfrom}-\mathsf{slack}(v_{\swPfrom},v)\bigr\} + 1
                \le
                \lambda(e_{to}) + \delta_{\swPto}$,
            \ie the switch at $v$ is temporal after accounting for the propagated delay from $v_{\swPfrom}$ to $v$.
    \end{description}
    We choose the earliest valid vertex on $\swPto$ as $v_{\swPto}$. %

    As $\sptree$ is a rooted tree, we process its edges in a root-to-leaf order such that when we compute $(v_{\swPto}, \swPfrom, \swPto)$, the parent switch vertex $v_{\swPfrom}$ is already fixed. Moreover, one can compute $v_{\swPto}$ from $v_{\swPfrom}$ in polynomial time by checking all candidate edges of $\swPto$ for \ref{item:switchoff_before_switchon} and \ref{item:svs_with_delaythis}.

    \bigparagraph{4. Correctness of the greedy construction.}
    Next we show that the \switchvertexset{} $\svset$ computed in this manner satisfies all three properties \ref{item:sptprime_is_sptstar}--\ref{item:svsprime_is_SVS_delaythis}. This gives the correctness of the greedy construction.
    \ref{item:sptprime_is_sptstar} holds by construction since $\svset$ is built by iterating over the edges of $\sptree$.
    \ref{item:svsprime_is_SVS_delaythis} holds, since the switch at $v_{\swPto}$ is temporal after applying the delay $\delta_{\swPto}$ by definition.
    It~remains to show \ref{item:vprime_before_vstar}, \ie  
    that $v_{\swPto} \le_{\swPto} v^*_{\swPto}$. 
    {Since $v_{\swPto}$ is chosen as the earliest valid vertex on $\swPto$, it suffices to show that $v^*_{\swPto}$ is also valid, \ie that it satisfies \ref{item:switchoff_before_switchon} and \ref{item:svs_with_delaythis}. Indeed, once $v^*_{\swPto}$ is shown to be valid, the greedy choice immediately yields $v_{\swPto} \le_{\swPto} v^*_{\swPto}$.}

    It follows from the definition of the \switchvertexset{s} that $v^*_{\swPto}$ occurs after $v^*_{\swPfrom}$ on $\swPfrom$. By the inductive hypothesis, $v_{\swPfrom} \le_{\swPfrom} v^*_{\swPfrom}$, and therefore $v_{\swPfrom} \le_{\swPfrom} v^*_{\swPto}$. Hence $v^*_{\swPto}$ satisfies \ref{item:switchoff_before_switchon}.

    To obtain \ref{item:svs_with_delaythis}, we show that $v^*_{\swPto}$ is reached in $\gcal'$ not later than in $\gcal^*$.
    Recall that in $\gcal^*$, the only delayed edge on $\swPfrom$ is the temporal edge $e^*_{\swPfrom}$ going out of $v^*_{\swPfrom}$, delayed by $\delta_{\swPfrom}$.
    This delay is propagated to all the edges after $e^*_{\swPfrom}$ and results in reaching $v^*_{\swPto}$ at some time $t$ such that $t< t^*_{\swPto}+\delta_{\swPto}$, since $v^*_{\swPto}$ is a switch-vertex in $\gcal^*$. 
    {Now, the greedy choice places the delayed edge $e_{\swPto}$ no later on $\swPto$ than the corresponding edge $e^*_{\swPto}$ in the optimal realization. Delaying $e_{\swPto}$ by $\delta_{\swPto}$ therefore propagates to $e^*_{\swPto}$ with a delay of at most $\delta_{\swPto}$. Hence, in $\gcal'$, the vertex $v^*_{\swPto}$ is reached on $\swPfrom$ at some time $t'\le t< t^*_{\swPto}+\delta_{\swPto}$, and the switch at $v^*_{\swPto}$ is temporal. Therefore $v^*_{\swPto}$ satisfies \ref{item:svs_with_delaythis}.}

    As a result, $v^*_{\swPto}$ is a valid candidate for $v_{\swPto}$, and since $v_{\swPto}$ is chosen as the earliest valid vertex, this yields $v_{\swPto} \le_{\swPto} v^*_{\swPto}$.

    \bigparagraph{5. Running time.}
    We enumerate $\bigoh(k^k)$ switch-path-trees and, for each of them, at most $(b+1)^{k-1}$ possibilities for the guessed delays on the non-source base-paths. For every such guess, the greedy algorithm computes $\svset'$ in polynomial time as argued above. Therefore the total running time is
    \[
        \bigoh\!\left(k^k\cdot (b+1)^{k-1}\cdot |\gcal|^{\bigoh(1)}\right)
        = \bigoh\!\left(b^k\cdot k^k\cdot |\gcal|^{\bigoh(1)}\right),
    \]
    and \dprm is \FPT{} parameterized by $b+k$.
\end{proof}
\fi

    The previous algorithm relies on the fact that the earliest viable switch along a path both maximizes the reachability of $s$ and minimizes the additional cost resulting from propagation. This simple argument fails when \emph{advances} are allowed, because an advance propagates \emph{backward} along the base-path $P$ and the advances of multiple switches off $P$ can interact in multiple ways. 
    As a result, the two objectives no longer align: whether a given switch vertex is valid may depend on advances caused by switches that are assigned only later, creating a circular dependency.
    
    We resolve this by refining the idea behind the ILP in \Cref{thm:budgeted_XP_by_k} and guessing more values witnessed by an optimal solution:
    We guess for each non-root base-path $Q$ six values capturing the delay and advance propagations at its switch vertex on the parent path, the active delay applied there, the advance arriving from $Q$'s children, and the label difference of the two incident temporal edges. We additionally guess a child-ordering $\sigma$ for each base-path.
    Once these values are fixed, the guessed propagations impose \emph{slack-distance conditions} between consecutive switch vertices on each parent path: either an exact condition when propagation is expected from another switch vertex (fixing both a lower and upper bound on the slack, forcing the switch vertex onto a short zero-slack subpath) or a minimum condition when no propagation is expected (lower bound only). Maximal runs of siblings linked by exact conditions form \emph{batches} whose relative positions are determined once the first member is placed.
    We then run a greedy construction root-to-leaves, placing each batch-head at the earliest position where 1) the minimum distance condition from the previous switch is satisfied and 2) the rest of the batch can be aligned from there such that slack-distance-conditions are satisfied. An exchange argument shows the first such match never destroys the existence of an optimal realization consistent with the guessed values.
    This argument utilizes an observation that all possible switches between two base-paths for which the incident edges have a specific label difference must have the same order along both base-path; hence earlier switches along $Q$ are also earlier along the base-path that we want to switch to.
    Putting this all together, we get an \FPT algorithm with running time $\bigoh\!\left(b^{6k}\cdot k^{2k}\cdot |\gcal|^{\bigoh(1)}\right)$.

\begin{theorem}\label{thm:all_fpt_by_b_and_k}
    \dasprm is in \FPT when parameterized by $b+k$, and can be solved in $\bigoh\!\left(b^{6k}\cdot k^{2k}\cdot |\gcal|^{\bigoh(1)}\right)$ time.
\end{theorem}
\iflong
\begin{proof}
Let $\gcal= \bigcup_{i\in [k]}P_i$ be the input \kpathgraph{k} and let $s$ be a source.
Furthermore, let $\gcal^*$ be any optimal solution for \dasprm, that is, a \kpathgraph{k} maximizing $|\rcal^{\gcal^*}(s)|$ among all \kpathgraph{k}s that can be  obtained from $\gcal$ by shifting / delaying / advancing within the budget $b$.
We are going to describe the algorithm for the general shifting setting and then describe how to adapt it for the delay-only and advance-only settings.

By \Cref{lem:switchverts_is_enough}, there is an \switchvertexset{} $\svset^*$ such that $\rcal^{\gcal^*}_{\svset^*}(s)=\rcal^{\gcal^*}(s)$ and by definition it implies an \switchpathtree{} $\sptree_{\svset^*}$.
By \Cref{lem:spt_enumeration} we can enumerate all possibilities for $\sptree_{\svset^*}$ in $\bigoh(k^{k})$ time.
Fix one such $\sptree$ and assume that $\sptree=\sptree_{\svset^*}$.
Additionally, we guess a permutation $\sigma$ of the $k-1$ non-root paths which will tell us in which order multiple children of the same base-path leave that base-path, contributing at most $(k-1)!$ additional possibilities.
The argument is structured into 4 steps.

\bigparagraph{1. Canonical form of an optimal solution.}
    By \Cref{lem:canonical_shift_form}, we may assume that the optimal solution applies at each switch $(v_Q, P, Q)$ exactly two local operations: an active advance $\actAdv{Q} \le 0$ on $e_P^-$ (the temporal edge of $P$ entering $v_Q$), and an active delay $\delQ{Q} \ge 0$ on $e_Q^+$ (the temporal edge of $Q$ leaving $v_Q$). No other edge is modified.

    Since only $e_P^-$ is modified at $v_Q$ and it lies to the left of $v_Q$, the delay propagating rightward along $P$ passes through $v_Q$ unchanged: $\DelIn{Q} = \DelOut{Q}$.
    The advance propagating leftward from $v_Q$ is the sum of the active advance and whatever advance arrives from the right:
    \[
        \AdvOut{Q} = \actAdv{Q} + \AdvIn{Q}.
    \]
    By \Cref{cor:delay_and_advance_should_never_collide}, no segment of $P$ carries both a delay and an advance simultaneously, so these propagations are well-separated.
    Together, this determines the propagation structure between any two consecutive switch vertices on a base-path entirely from the guessed values at those vertices.

\bigparagraph{2. Enumeration of the relevant values.}
    As before, treat $(\emptyset,P_s)\in E(\sptree)$ as the empty switch onto the root path $P_s$.
    For each non-root base-path $Q$ with parent path $P$, we guess the following values at its switch vertex $v_Q$ (see \Cref{fig:fpt-k_b-figure} for an illustration):
    \begin{itemize}
        \item the active delay $\delQ{Q} \in [0,b]$ on the temporal edge of $Q$ leaving $v_Q$;
        \item the delay $\DelOut{Q} \in [0,b]$ propagated rightward through $v_Q$ on $P$ (equal to $\DelIn{Q}$ arriving from the left, by Part~1);
        \item the advance $\AdvOut{Q} \in [-b,0]$ sent leftward from $v_Q$ on $P$, and the advance $\AdvIn{Q} \in [-b,0]$ arriving at $v_Q$ from the right on $P$;
        \item the advance $\advInQ{Q} \in [-b,0]$ arriving at $v_Q$ from the right on $Q$ (from $Q$'s children in $\sptree$);
        \item the label difference $\labDiff{}{Q} \in [-b,b]$, where $\labDiff{}{Q} = \lambda(e_Q^+) - \lambda(e_P^-)$ is the difference between the label of the temporal edge $e_Q^+$ of $Q$ leaving $v_Q$ and the label of the temporal edge $e_P^-$ of $P$ entering $v_Q$.
    \end{itemize}
    The label difference $\labDiff{}{Q}$ is bounded by $b$: if $|\labDiff{}{Q}| > b$ the switch cannot be made temporal within the budget.
    By \Cref{cor:delay_and_advance_should_never_collide}, $\DelOut{Q} > 0$ implies $\AdvIn{Q} = 0$ and $\AdvOut{Q}<0$ implies $\DelIn{Q} = 0$.
    The active advance applied at $v_Q$ on $P$ can be computed from the guessed values as $\actAdv{Q} = \AdvOut{Q} - \AdvIn{Q} \in\ [-b,0]$,
    and the budget condition is
    $\sum_{Q} \Bigl(\delQ{Q} + \lvert\actAdv{Q}\rvert\Bigr) \leq b$.
    \begin{figure}[t]
        \centering
        \includegraphics[width=0.8\textwidth]{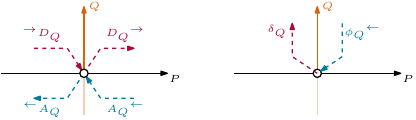}
        \caption{Illustration of the guessed values for a non-root base-path $Q$ with parent $P$. $\labDiff{}{Q}$ is the difference between the labels of the temporal edges $e_P$ and $e_Q$ entering and leaving the switch vertex $v_Q$. The active delay $\delQ{Q}$ is applied to $e_Q$; the active advance $\actAdv{Q}$ is derived as $\AdvOut{Q}-\AdvIn{Q}$. The values $\DelOut{Q}$, $\AdvOut{Q}$, $\AdvIn{Q}$ describe the delay/advance propagations on $P$ at $v_Q$, and $\advInQ{Q}$ the advance arriving at $v_Q$ from $Q$'s children.}\label{fig:fpt-k_b-figure}
    \end{figure}

    The number of possibilities is $(b+1)\cdot\bigl((b+1)^5(2b+1)\bigr)^{k-1}$. For the root $(\emptyset,P_s)$ only $\AdvIn{P_s}$ is non-trivial, giving $b+1$ choices; for each of the $k-1$ non-root paths there are $b+1$ choices each for $\delQ{Q}$, $\DelOut{Q}$, $\AdvOut{Q}$, $\AdvIn{Q}$, and $\advInQ{Q}$, and $2b+1$ choices for $\labDiff{}{Q}$.
    Since the enumeration considers all possible choices of these values, one must match the canonical optimal solution $\gcal^*$.

\bigparagraph{3. Greedy construction.}
    Assume we guessed \switchtuple{\swPfrom}{\swPto} for every  $(\swPfrom,\swPto)\in E(\sptree)$.
    We will iteratively compute an \switchvertexset $\svset$ satisfying the following properties:
    \begin{description}
        \item[\namedlabel{item:sptprime_is_sptstar}{(\switchvertexset 1)}(\switchvertexset 1)] \quad $\sptree_{\svset}=\sptree=\sptree_{\svset^*}$,
        \item[\namedlabel{item:vprime_before_vstar}{(\switchvertexset 2)}(\switchvertexset 2)] \quad for all $(v, \swPfrom,\swPto)\in \svset$ and the corresponding $(v^*, \swPfrom,\swPto)\in \svset^*$, the vertex $v$ is before (or equal to) $v^*$ on $\swPto$, and
        \item[\namedlabel{item:svsprime_is_SVS_delaythis}{(\switchvertexset 3)}(\switchvertexset 3)] \quad $\svset$ is a temporal \switchvertexset in the \kpathgraph{k} $\gcal'$ obtained from $\gcal$ by applying, for every $(v, \swPfrom, \swPto)\in \svset$, the delay $\delta_{\swPto}$ to $e_{\swPto}$.
    \end{description}
    We compute $\svset$ by going root-to-leaves along $\sptree$.

    As in the previous proof, we treat $(s,\emptyset,P_s)$ as the switch onto the root path $P_s$.
    Hence, in both \switchvertexset{s} we switch onto $P_s$ at $s$ with delay $\lDel{\emptyset}{P_s}=0$, and we can treat this as the base case.

    Next consider a non-source base-path $P$ and suppose we fixed the switch onto $P$ as $v_P$ which incurs the delay $\lDel{O}{P}\in[0,b]$ on $P$ and is reached by a propagated advance of $\propTo{O}{P}\in[-b,0]$ via $P$.
    Let $Q^1, Q^2, \dots,Q^\ell$ be the $\ell$ children of $P$ in \sptree ordered according to $\sigma$, \ie their respective switch vertices are on $P$ in the order $v_P \le_P v_{Q^1} \le_P v_{Q^2} \le_P \dots \le_P v_{Q^\ell}$.
    For each child path $Q^i$, we call a common vertex $v$ of $P$ and $Q^i$ \emph{valid} for the switch from $P$ to $Q^i$ if and only if
    \begin{description}
        \item[\namedlabel{item:valid-switch-order}{(P1)}(P1)] $v_{P}\le_{P} v$,
        \item[\namedlabel{item:valid-switch-label-difference}{(P2)}(P2)] $\lambda(e_{Q^i}^+(v))-\lambda(e_P^-(v))=\labDiff{P}{Q^i}$, and
        \item[\namedlabel{item:valid-switch-temporality}{(P3)}(P3)]
            $\lambda(e_{P}^-(v)) + \lAdv{P}{Q^i} + \propFrom{P}{Q^i} + 1
            \leq
            \lambda(e_{Q^i}^+(v)) + \lDel{P}{Q^i} + \propTo{P}{Q^i}$, \ie the switch at $v$ is temporal.
    \end{description}
    Property~\ref{item:valid-switch-label-difference} strongly restricts which common vertices can be valid. In particular, the valid candidates occur in the same order on~$P$ and~$Q^i$. See \Cref{fig:fpt-kb-monotone} for an illustration of this monotonicity, which we formalize in the following claim.
    \begin{figure}[t]
        \centering
        \includegraphics[width=0.9\textwidth]{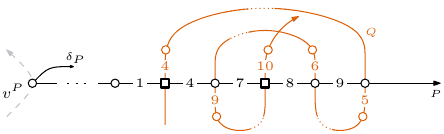}
        \caption{Valid switches (squares) for a switch edge $(P,Q)\in E(\sptree)$ with $\labDiff{P}{Q}=3$ occur in the same order on $P$ and $Q$.}\label{fig:fpt-kb-monotone}
    \end{figure}
    \begin{claim}[valid switches are monotone]
    Let $x$ and $y$ be two valid vertices for the switch from $P$ to $Q^i$. Then
    $x \le_{P} y
        \ \Leftrightarrow\
        x \le_{Q^i} y$.
    That is, the valid vertices for any switch occur in the same order on both paths.
    \end{claim}
    \begin{claimproof}
    Assume that $x \le_{P} y$. Since $P$ is a temporal path, its labels are increasing along the path, and therefore $\lambda(e_P^-(x)) \le \lambda(e_P^-(y))$.
    As both $x$ and $y$ are valid, property~\ref{item:valid-switch-label-difference} yields $\lambda(e_{Q^i}^+(x))-\lambda(e_P^-(x)) =
        \labDiff{P}{Q^i}
        =
        \lambda(e_{Q^i}^+(y))-\lambda(e_P^-(y))$.
    Hence
        $\lambda(e_{Q^i}^+(x))
        = \labDiff{P}{Q^i}+
        \lambda(e_P^-(x))
        \leq\labDiff{P}{Q^i}+
        \lambda(e_P^-(y))
        =
        \lambda(e_{Q^i}^+(y))$.
    Since $Q^i$ is again a temporal path, its labels are increasing along the path, so $\lambda(e_{Q^i}^+(x)) \le \lambda(e_{Q^i}^+(y))$ implies $x \le_{Q^i} y$.
    The converse direction follows by exchanging the roles of $P$ and $Q^i$.
    \end{claimproof}

    Next, we compute the \emph{slack-distance} conditions between the children and with respect to the parent switch.
    For consecutive switch vertices $v_{Q^i}$ and $v_{Q^{i+1}}$ on $P$, the propagations in the segment $[v_{Q^i}, v_{Q^{i+1}}]$ must be consistent with the guessed values: delay absorbed by the slack reduces $\DelOut{Q^i}$ to $\DelOut{Q^{i+1}}$, and advance absorbed reduces $\lvert\AdvOut{Q^{i+1}}\rvert$ to $\lvert\AdvIn{Q^i}\rvert$.
    By \Cref{cor:delay_and_advance_should_never_collide}, delay and advance cannot coexist in the same segment, so at most one of these reductions is nonzero.
    This gives the unified slack condition:
    \begin{equation}\label{eq:slack-condition}
        \bigl(\DelOut{Q^i} - \DelOut{Q^{i+1}}\bigr) + \bigl(\lvert\AdvOut{Q^{i+1}}\rvert - \lvert\AdvIn{Q^i}\rvert\bigr)
        \;\leq\;
        \mathsf{slack}(v_{Q^i}, v_{Q^{i+1}})
        \;\leq\;
        \DelOut{Q^i} + \lvert\AdvOut{Q^{i+1}}\rvert.
    \end{equation}
    The upper bound is active precisely when $\DelOut{Q^{i+1}} > 0$ or $\lvert\AdvIn{Q^i}\rvert > 0$ (propagation is expected to arrive at the far end), in which case it coincides with the lower bound — making the condition \textbf{exact}: the slack between $v_{Q^i}$ and $v_{Q^{i+1}}$ is uniquely fixed, and $v_{Q^{i+1}}$ must lie on the \emph{zero-slack subpath} reached after absorbing exactly that much slack from $v_{Q^i}$.
    Otherwise ($\DelOut{Q^{i+1}} = 0$ and $\lvert\AdvIn{Q^i}\rvert = 0$) only the lower bound applies — a \textbf{minimum distance} condition: $v_{Q^{i+1}}$ may be placed anywhere past the minimum, on any zero-slack subpath beyond that point.
    A maximal run of siblings linked by exact conditions forms a \emph{batch}; the zero-slack subpath of the leftmost member pins all subsequent positions in the batch exactly.

    For the parent-child pair $(v_P, v_{Q^1})$, \eqref{eq:slack-condition} applies with $\DelOut{Q^0} := \delQ{P}$ (the delay sent by $v_P$ on $P$) and $\AdvOut{Q^1}$ as above, but the \textbf{upper bound is always dropped}: since the greedy places $v_P$ before its children, $v_P$ may land earlier than in $\svset^*$, creating more slack and delivering \emph{less} propagation to $v_{Q^1}$ than guessed.
    Less delay arriving at $v_{Q^1}$ only eases its temporality condition; less advance arriving back at $v_P$ means less advance propagates further left, which is also acceptable.
    Hence all parent-child conditions are \textbf{minimum distance} conditions.

    Finally, we describe how to place the children of $P$ based on the fixed switch vertex $v_P$ and the distance conditions.
    We process the batches (as defined in \eqref{eq:slack-condition}) left-to-right according to $\sigma$.
    For each batch $\langle Q^i, \dots, Q^j\rangle$, we place $v_{Q^i}$ at the \emph{earliest} vertex on $P$ such that (i) the minimum-distance condition relative to the \emph{preceding} placement is satisfied: $v_P$ if $i = 1$ (the first batch), otherwise $v_{Q^{i-1}}$ (the last member of the previous batch), and (ii) every member of the batch $Q^m$, $i \le m \le j$, can be put on a valid vertex.
    If there is no valid positioning for the batch, then the current guessing branch is discarded as infeasible.

\bigparagraph{4. Correctness.}
    We show that $\svset$ satisfies all three properties~\ref{item:sptprime_is_sptstar}--\ref{item:svsprime_is_SVS_delaythis}.

    \ref{item:sptprime_is_sptstar} holds by construction.

    \ref{item:svsprime_is_SVS_delaythis}: batch-heads are placed at the earliest valid position, so temporal by definition; within-batch members are placed at exact slack from their predecessor, and since the actual propagation arriving there is at most the guessed value in the favorable direction (less delay / less advance than guessed only makes the temporality condition easier to satisfy), each such switch is temporal as well.

    It remains to show~\ref{item:vprime_before_vstar}, i.e., $v_{Q} \le_{Q} v^*_{Q}$ for every $(P,Q)\in E(\sptree)$.
    We proceed by induction over the \emph{greedy processing order}: the order in which the greedy places switch vertices, which visits paths root-first and processes children of each path in the order given by~$\sigma$.
    This order ensures that whenever $Q^i$ is processed, its anchor (either $v_P$ or $v_{Q^{i-1}}$) has already been placed and satisfies the inductive hypothesis.
    The root path $P_s$ is switched at $s$ in both $\svset$ and $\svset^*$, so the hypothesis holds trivially.

    Assume $v_Q \le v^*_Q$ holds for every path $Q$ processed before $Q^i$.
    Let $Q^i$ be the current path (child of $P$), and let $a$ denote its anchor ($a = v_P$ if $i=1$, else $a = v_{Q^{i-1}}$). By the inductive hypothesis, $a \le_P a^*$.

    In $\svset^*$, the guessed propagation values are realized exactly, so $v^*_{Q^i}$ satisfies the exact distance condition relative to $a^*$:
    $\mathsf{slack}(a^*,\, v^*_{Q^i}) = \text{threshold}_{Q^i}$, where $\text{threshold}_{Q^i}$ is determined solely by the guessed values.
    Since $a \le_P a^*$ and cumulative slack is additive along $P$, we get
    \[
        \mathsf{slack}(a,\, v^*_{Q^i})
        = \mathsf{slack}(a, a^*) + \mathsf{slack}(a^*,\, v^*_{Q^i})
        \ge \mathsf{slack}(a^*,\, v^*_{Q^i})
        = \text{threshold}_{Q^i}.
    \]
    Hence $v^*_{Q^i}$ satisfies the \textbf{minimum}-distance condition from $a$ (which is all we require in our construction).

    If $Q^i$ is a batch-head, our greedy construction places $v_{Q^i}$ at the earliest position satisfying two conditions:
        (i) $\mathsf{slack}(a, v_{Q^i}) \ge \text{threshold}_{Q^i}$, and
        (ii) the entire batch $\langle Q^i, \dots, Q^j \rangle$ can be placed at valid switch vertices using the exact slack conditions within the batch.
        We argue that $v^*_{Q^i}$ satisfies both: (i) was shown above; (ii) holds because in $\svset^*$ the full batch is placed according to the exact slack conditions starting from $v^*_{Q^i}$, with each subsequent member $v^*_{Q^m}$ on some valid position of the correct zero-slack subpath.
        Hence at $v^*_{Q^i}$ there exists a valid playing of the entire batch and (ii) is satisfied.

    If instead $Q^i$ is within a batch, then our greedy construction places $v_{Q^i}$ at the earliest valid position within the zero-slack subpath satisfying $\mathsf{slack}(v_{Q^{i-1}}, v_{Q^i}) = T$ with $T= \text{threshold}_{Q^i}$. Setting $S := \mathsf{slack}(v_{Q^{i-1}},\, v^*_{Q^{i-1}}) \ge 0$ (which is non-negative since $v_{Q^{i-1}} \le_P v^*_{Q^{i-1}}$ by IH), we get
    \[
        \mathsf{slack}(v_{Q^{i-1}},\, v^*_{Q^i})
        = S + \mathsf{slack}(v^*_{Q^{i-1}},\, v^*_{Q^i})
        = S + T \ge T.
    \]
    Since $\mathsf{slack}(v_{Q^{i-1}}, \cdot)$ is non-decreasing and equals $T$ at $v_{Q^i}$ while it equals $S + T \ge T$ at $v^*_{Q^i}$, monotonicity gives $v_{Q^i} \le_P v^*_{Q^i}$.

    In both cases $v_{Q^i} \le_P v^*_{Q^i}$, completing the induction and establishing~\ref{item:vprime_before_vstar}.

\bigparagraph{5. Running time.}
    We enumerate $\bigoh(k^k)$ switch-path-trees, at most $(k-1)!$ child-orderings per tree, and for each such combination
    \[
        (b+1)\cdot \bigl((b+1)^5(2b+1)\bigr)^{k-1}
    \]
    tuples of guessed values. For every fixed guess, the greedy procedure only scans the relevant portions of the involved base-paths, computes cumulative slack values, and checks the defining equalities for admissibility and validity of switch vertices. Hence each single guess can be processed in time polynomial in $|\gcal|$.

    Therefore the total running time is
    \[
        \bigoh\!\left(k^k \cdot (k-1)! \cdot (b+1)\cdot \bigl((b+1)^5(2b+1)\bigr)^{k-1}\cdot |\gcal|^{\bigoh(1)}\right)
        = \bigoh\!\left(b^{6k}\cdot k^{2k}\cdot |\gcal|^{\bigoh(1)}\right).
    \]
    Hence \dasprm is \FPT{} parameterized by $b+k$.
\end{proof}
\fi

\section{Discussion}
The most obvious open question of our work is to settle the complexity of \sprm and \aprm when parameterized by $k$. 
Is the problem fixed-parameter-tractable, or is it \Wone-hard? An algorithm that utilizes switch-path trees would need to be able to discard hard \switchpathtree-choices as sub-optimal, or use a completely different strategy. Towards a hardness result, a novel reduction idea would be needed.
A different direction would be to consider the problem with more than one sources. In this case, one can define several different objectives to optimize, obvious examples being maximization of the union of reachable vertices from all sources and the maximization of the minimum number of reachable vertices from any source. In the former case, it is easy to observe that we can add an auxiliary source $s^*$ that is connected via one base-path to all sources and thus reduce the problem to \sprm.%

From a broader perspective, our results can be seen as an extension of the recent positive results on temporal \kpathgraphs{k}~\cite{doring_temporalconnected_2025}, this time for the problem of maximum reachability. 
Thus, they serve as an indication that there exist ``natural'' classes of temporal graphs that can enable fixed-parameter tractability. This contrasts the vast of the literature on temporal graphs where ``everything becomes immediately hard''. Are there other natural classes that admit fixed parameter algorithms? Is there a subclass of temporal $k$-path graphs that lets us overcome the current intractability results for \sprm and other problems? Answering, even partially, these questions can provide fertile ground for future research.

\bibliography{zotero}

\end{document}